\documentclass[12pt]{article}

\usepackage[margin=1.1in]{geometry}

\usepackage{amsfonts}
\usepackage{amsmath,amssymb}
\usepackage{graphicx,amscd,xypic}
\usepackage{bbding}
\usepackage{indentfirst}
\usepackage{mathrsfs, dsfont, bbm}
\usepackage[T1]{fontenc}
\usepackage{color}

\newtheorem{thm}{Theorem}
\newtheorem{defi}{Definition}
\newtheorem{prop}{Proposition}
\newtheorem{coro}{Corollary}
\newtheorem{lemma}{Lemma}
\newtheorem{assmp}{Assumption}
\newtheorem{example}{Example}
\newtheorem{remark}{Remark}
\newenvironment{proof}{\hspace{0ex}\textsc{Proof}.\hspace{1ex}}{\hfill$\Box$  \noindent}
\DeclareMathOperator{\dd}{\mathrm{d\!}}
\DeclareMathOperator{\BA}{ \mathbf{A}}
\DeclareMathOperator{\BB}{ \mathbf{B}}
\DeclareMathOperator{\BC}{ \mathcal{C}}
\DeclareMathOperator{\BCC}{ \mathbf{C}}
\DeclareMathOperator{\BD}{\mathbf{D}}
\DeclareMathOperator{\BE}{\mathbf{E}}
\DeclareMathOperator{\BF}{\mathcal{F}}
\DeclareMathOperator{\BP}{\mathbf{P}}
 
\DeclareMathOperator{\R}{\mathbb{R}}

\DeclareMathOperator{\setd}{\mathfrak{D}}
\DeclareMathOperator{\setdp}{{\mathfrak{D}}^{+}}
\DeclareMathOperator{\sets}{\mathfrak{S}}

\DeclareMathOperator*{\argmin}{\mathrm{argmin}}
\DeclareMathOperator{\Var}{\mathbf{Var}}
\begin{document}

\title{Continuous-Time Mean-Variance Portfolio Selection \\
        with Constraints on Wealth and Portfolio}
\author{
Xun Li\thanks{Department of Applied Mathematics, The Hong Kong Polytechnic University, Hong Kong, China. This author acknowledges financial supports from Research Grants Council of Hong Kong under grants No. 520412 and 15209614. E-mail: malixun@polyu.edu.hk.} \quad and \quad Zuo Quan Xu\thanks{Corresponding author. Department of Applied Mathematics, The Hong Kong Polytechnic University, Hong Kong, China. This author acknowledges financial supports from Hong Kong Early Career Scheme (No. 533112) and Hong Kong General
Research Fund (No. 529711). E-mail: maxu@polyu.edu.hk.}}


\bibliographystyle{plain}

\maketitle

\begin{abstract}
We consider continuous-time mean-variance portfolio selection with
bankruptcy prohibition under convex cone portfolio constraints.
This is a long-standing and
difficult problem not only because of  its theoretical significance,
but also for its practical importance.
First of all, we transform the above problem into an equivalent
mean-variance problem with bankruptcy prohibition \textit{without}
portfolio constraints. The latter is then treated using martingale
theory. Our findings indicate that we can directly present the
semi-analytical expressions of the pre-committed efficient
mean-variance policy without a viscosity solution technique but within a
general framework of the cone portfolio constraints. The numerical
simulation also sheds light on results established in this
paper. \\ [12pt]
\noindent\textit{Keywords:}  continuous-time,
mean-variance portfolio selection, bankruptcy prohibition, convex
cone  constraints,
efficient frontier, 
HJB equation
\end{abstract}


\section{Introduction}
\noindent
Since Markowitz \cite{Markowitz} published his seminal
work on the mean-variance portfolio selection sixty years ago, the
mean-risk portfolio selection framework has become one of the most
prominent topics in quantitative finance. Recently, there has
been increasing interest in studying the dynamic mean-variance portfolio
problem with various constraints, as well as addressing their financial
applications. Typical contributions include
Bielecki, Jin, Pliska and Zhou \cite{BJPZ}, Cui, Gao, Li and Li \cite{CuiGaoLiLi2012},
Cui, Li and Li \cite{CuiLiLi2012},
Czichowsky and Schweizer \cite{CzichowskySchweizer2012}, Heunis \cite{Heunis}, Hu and Zhou \cite{HZ},
Duffie and Richardson \cite{DR},Labb\'e and Heunis \cite{LabbeHeunis},
Li and Ng \cite{LiNg}, Zhou and Li \cite{ZhouLi}, and Li, Zhou and Lim \cite{LZL}.
The dynamic mean-variance problem can be treated in a forward-looking way by starting
with the initial state. In some financial engineering problems, however, one
needs to study stochastic systems with constrained conditions, such as cone-constrained
policies. This naturally results in a continuous-time mean-variance
portfolio selection problem with \textit{constraints for the wealth process}, (see \cite{BJPZ}),
and \textit{constraints for the policies}, (see \cite{HZ} and \cite{LZL}),
for a given expected terminal target.
To the best of our knowledge, despite active research efforts in this direction in
recent years, there has been little progress in the continuous-time mean-variance problem with the
\textit{mixed} restriction of both bankruptcy prohibition and convex
cone portfolio constraints.
This paper aims to address this long-standing and notoriously
difficult problem, not only for its theoretical significance,
but also for its practical importance.
New ideas, significantly different from those developed in the existing literature,
establish a general theory for stochastic control problems with mixed
constraints for state and control variables.

Li, Zhou and Lim \cite{LZL} considered a continuous-time mean-variance problem with no-shorting constraints,
while Cui, Gao, Li and Li \cite{CuiGaoLiLi2012}
developed its counterpart in discrete-time. Bielecki,
Jin, Pliska and Zhou \cite{BJPZ} paved the way for investigating
continuous-time mean-variance with bankruptcy prohibition using a
martingale approach. Sun and Wang \cite{SunWang} introduced a market
consisting of a riskless asset and one risky portfolio under
constraints such as market incompleteness, no-shorting, or partial
information.
Labb\'e and Heunis \cite{LabbeHeunis} 
employed a duality method to analyze both the mean-variance
portfolio selection and mean-variance hedging problems with general
convex constraints. In particular, Czichowsky and Schweizer
\cite{CzichowskySchweizer2012} further studied cone-constrained
continuous-time mean-variance portfolio selection problem with the price
processes being semimartingales. Cui, Gao, Li and Li
\cite{CuiGaoLiLi2012} and Cui, Li and Li \cite{CuiLiLi2012} derived
explicit optimal semi-analytical mean-variance policies for
discrete-time markets under no-shorting and convex cone constraints,
respectively. Also, F\"ollmer and Schied \cite{FollmerSchied2004}
and Pham and Touzi \cite{PhamTouzi1999} showed that in a constrained
market, no arbitrage opportunity is equivalent to the existence of a
supermartingale measure, under which the discounted wealth process
of any admissible policy is a supermartingale, (see Carassus, Pham
and Touzi \cite{CarassusPhamTouzi} for a situation with upper bounds
on proportion positions). In particular, Xu and Shreve in their
two-part papers \cite{XuShreve1,XuShreve2} investigated a utility
maximization problem with no-shorting constraints using duality
analysis. Recently, Heunis \cite{Heunis} carefully considered to minimize the expected value of a
general quadratic loss function of the wealth in a more general setting where
there is a specified convex constraint on the portfolio over the
trading interval, together with a specified almost-sure lower-bound
on the wealth at close of trade.

The existing theory and methods cannot directly handle the
continuous-time mean-variance problem with the mixed restriction of
both bankruptcy prohibition and convex cone portfolio constraints.
Therefore, we introduce a Hamilton-Jacobi-Bellman (HJB) equation to
analyze this problem. Based on our analysis, we find out that the
market price of risk in policy is actually independent of the wealth
process. This important finding allows us to overcome the difficulty
of the original problem and also makes the similar continuous-time
financial investment problem both interesting and practical. Hence,
we can construct a transformation to tackle the model presented
above using linear-quadratic convex optimization technique. Finally,
we discuss the equivalent problem using the model of Bielecki, Jin,
Pliska and Zhou \cite{BJPZ} without the viscosity solution technique
developed in \cite{LZL}.

This paper is organized as follows. In Section 2,
we formulate the continuous-time mean-variance problem with the mixed restriction of a bankruptcy prohibition and convex cone portfolio constraints.
In Section 3, we transform our mean-variance problem
into an equivalent mean-variance problem with a bankruptcy prohibition \textit{without} convex cone portfolio constraints.
We derive the pre-committed continuous-time efficient
mean-variance policy for our problem using the model proposed by Bielecki,
Jin, Pliska and Zhou \cite{BJPZ} in Section 4.
In Section 5, we discuss properties of the continuous-time mean-variance problems with different constraints.
In Section 6, we present a numerical simulation to illustrate results established in this paper.
Finally, we summarize and conclude our work in Section 7.


\section{Problem Formulation and Preliminaries}

\subsection{Notation}
We make use of the following notation: \\[10pt]
\begin{tabular}{rcl}
$z^+$ & : & the transformation of vector $z$ with every component $z_i^+ = \max\{z_i, 0\}$; \\
$z^-$ & : & the transformation of vector $z$ with every component $z_i^- = \max\{-z_i, 0\}$; \\
$M'$ & : & the transpose of any matrix or vector $M$; \\
$\Vert M\Vert$ & : & $\sqrt{\sum_{i,j}m_{ij}^2}$ for any matrix or vector $M = (m_{ij})$; \\
$\mathbb{R}^m$ & : & $m$ dimensional real Euclidean space; \\
$\mathbb{R}_+^m$ & : & the subset of  $\mathbb{R}^m$ consisting of elements with nonnegative components.
\end{tabular}
\vskip 12pt
The underlying
uncertainty is generated on a fixed filtered complete probability space
$(\Omega, \mathbf{F}, \BP, \{ {\BF}_t \}_{t \geqslant  0})$ on which is defined a
standard $\{ {\BF}_t \}_{t \geqslant  0}$-adapted $m$-dimensional Brownian
motion $W(\cdot) \equiv (W^1(\cdot), \cdots, W^m(\cdot))'$.
Given a probability space $(\Omega, \mathbf{F}, \BP)$ with a filtration
$\{{\BF}_t \vert a \leqslant  t \leqslant  b\} (-\infty <  a < b \leqslant  +\infty)$ and
a Hilbert space $\mathcal{H}$
with the norm $\Vert\cdot\Vert_{\mathcal{H}} $, we can
define a Banach space
\begin{eqnarray*}
L_{\BF}^2(a, b; \mathcal{H}) = \left\{\varphi(\cdot) \bigg\vert
\begin{array}{l}
\varphi(\cdot) \mbox{ is an } {\BF}_t\mbox{-adapted, }
\mathcal{H}\mbox{-valued measurable }\\
\mbox{process on } [a, b]  \mbox{ and } \Vert\varphi(\cdot)\Vert_{\BF}<+\infty
\end{array}
\right\}
\end{eqnarray*}
with the norm
\begin{equation*}
\Vert\varphi(\cdot)\Vert_{\BF}
= \left(\BE\left[\int_a^b\Vert\varphi(t, \omega)\Vert_{\mathcal{H}} ^2\dd t\right]\right)^{\frac{1}{2}}.
\end{equation*}

\subsection{Problem Formulation}
\noindent
Consider an arbitrage-free financial market where $m+1$ assets are
traded continuously on a finite horizon $[0, T]$.
One asset is a \textit{bond}, whose price $S_0(t),\;t\geqslant  0$, evolves
according to the differential equation
\begin{equation*}
\begin{cases}
\dd S_0(t)   =   r(t)S_0(t)\dd t, \quad t \in [0, T], \\
 S_0(0)   =  s_0 > 0,
\end{cases}
\end{equation*}
where $r(t) $ is the interest rate of the bond at time $t$. The remaining $m$ assets are
\textit{stocks}, and their prices are modeled by the system of stochastic differential equations
\begin{equation*}
\begin{cases}
\dd S_i(t) = S_i(t)\{b_i(t)\dd t + \sum_{j=1}^m\sigma_{ij}(t)\dd W^j(t)\}, \quad t \in [0, T], \\
 S_i(0) = s_i > 0,
\end{cases}
\end{equation*}
where $b_i(t) $ is the appreciation rate of the $i$-th stock and
$\sigma_{ij}(t)$ is the volatility coefficient at time $t$.
Denote $b(t):=(b_1(t),\cdots,b_m(t))'$ and $\sigma(t):=(\sigma_{ij}(t))$.
We assume throughout that $r(t)$, $b(t)$ and $\sigma(t)$ are given deterministic, measurable, and uniformly bounded functions on $[0, T]$. In addition, we assume that the non-degeneracy condition on $\sigma(\cdot)$, that is,
\begin{equation}\label{sigmabounded}
y'\sigma(t)\sigma(t)'y \geqslant  \delta y'y, \quad \forall\; (t,y)\in[0,T]\times\R^m,
\end{equation}
is satisfied for some scalar $\delta>0$.  Also, we define the
{excess return vector}
\begin{equation*}
B(t) =  b(t) - r(t)\mathbf{1},
\end{equation*}
where $\mathbf{1}=(1,1,\cdots,1)'$ is an $m$-dimensional   vector.
\par
Suppose an agent has an initial wealth $x_0>0$ and the total wealth of his position at time $t \geqslant  0$ is $X(t)$. Denote by $\pi_i(t)$, $ i = 1, \cdots, m, $ the total market value of the agent's wealth  in the $i^{\mathrm{th}}$  stock at time $t$. We call $\pi(\cdot):=(\pi_1(\cdot),\cdots,\pi_m(\cdot))'\in  L_{\BF}^2(0,T;\R^m)$ a portfolio.
We will consider self-financing portfolios in this paper.
Then it is well-known that $X(t),\;t\geqslant 0$, follows (see, e.g., \cite{ZhouLi})
\begin{equation}\label{eq:state-positive}
\begin{cases}
\dd X(t) = [r(t)X(t) + \pi(t)'B(t)]\dd t + \pi(t)'\sigma(t)\dd W(t),  \\
 X(0) =x_0.
\end{cases}
\end{equation}
\par
An important restriction considered in this paper is the convex cone portfolio constraints, that is $\pi(\cdot)\in\BC$, where
\begin{align}\label{defi:C}
\BC=\big\{\pi(\cdot): \pi(\cdot)\in L_{\BF}^2(0,T;\R^m), \; C(t)'\pi(t) \in \mathbb{R}_+^k, \forall\;
t\in[0,T]\big\},
\end{align}
and $C(\cdot):[0,T]\mapsto \R^{m\times k}$ is a given deterministic and measurable function.  
Another important yet relevant restriction considered in this paper is the prohibition of bankruptcy, i.e., we required that
\begin{align}\label{nobankrupt}
X(t) \geqslant 0,\quad \forall\;t\in[0,T].
\end{align}
On the other hand, borrowing from the money market (at the interest rate $r(\cdot)$) is still allowed; that is, the money invested in the bond $\pi_0(\cdot)=X(\cdot)-\sum_{i=1}^m\pi_m(\cdot)$ has no constraint.
\begin{defi}\label{defi:control-positive}
A portfolio $\pi(\cdot )$ is called an admissible control (or portfolio) if $\pi(\cdot) \in \BC$ and the corresponding wealth process  $X(\cdot)$ defined in \eqref{eq:state-positive} satisfies \eqref{nobankrupt}. In this case, the process  $X(\cdot)$   is called an admissible wealth process, and $ (X(\cdot), \pi(\cdot ))$ is called an admissible pair.
\end{defi}
\begin{remark}
In view of the boundedness of $\sigma(\cdot)$ and the non-degeneracy condition \eqref{sigmabounded}, we have that
$\pi (\cdot)\in L_{\cal F}^2(a, b; \mathbb{R}^m)$ if and only if $\sigma(\cdot)'\pi(\cdot) \in L_{\cal F}^2(a, b; \mathbb{R}^m)$. The latter is often used to define the admissible process in the literature, for instance, Bielecki, Jin, Pliska and Zhou \cite{BJPZ}.
\end{remark}
\begin{remark}\label{remark:interval}
It is easy to show that both the set of all admissible controls and the set of all admissible wealth processes are convex.
In particular, the set of expected terminal wealths \[\big\{\BE[X(T)]:\text{$X(\cdot)$ is an admissible process}\big\}\]
is an interval.
\end{remark}
\par
Mean-variance portfolio selection refers to the problem of, given a favorable mean level $d$,  finding an
allowable investment policy, (i.e., a dynamic portfolio satisfying all the constraints),
such that the expected terminal wealth  $\BE[X(T)]$ is $d$
while the risk measured by the variance of the terminal wealth
\begin{align*}
\Var(X(T)) = \BE[X(T) - \BE[X(T)]]^2 = \BE[X(T) - d]^2
\end{align*}
is minimized.
\par
We impose throughout this paper the following assumption:
\begin{assmp}\label{ass:d}
The value of the expected terminal wealth $d$ satisfies
$d \geqslant x_0e^{\int_0^T r(s)\dd s}$.
\end{assmp}
\begin{remark}\label{rem:d}
Assumption \ref{ass:d} states that the investor's expected terminal wealth $d$ should be no less
than $x_0e^{\int_0^T r(s)ds}$ which coincides with the amount that he/she would earn if
all of the initial wealth is invested in the bond for the entire investment period. Clearly, this is a reasonable
assumption, for the solution of the problem under $d < x_0e^{\int_0^T r(s)ds}$
is foolish for rational investors.
\end{remark}
\begin{defi}
The mean-variance portfolio selection problem is formulated as the
following optimization problem parameterized by $d $:
\begin{eqnarray}\label{eq:utility-positive}
\begin{array}{ll}
       \min & \quad\Var (X(T))= \BE[X(T) - d]^2, \\ [2mm]
         \mathrm{subject\ to} &\quad
       \left\{\begin{array}{l}
       \BE[X(T)] = d, \\ [2mm]
       \pi(\cdot) \in \BC \mbox{ and } X(\cdot) \geqslant  0, \\ [2mm]
       (X(\cdot), \pi(\cdot )) ~ \hbox{\rm satisfies equation \eqref{eq:state-positive}}.
       \end{array}\right.
\end{array}
\end{eqnarray}
An optimal control satisfying \eqref{eq:utility-positive} is called
an efficient strategy, and $(\Var(X(T)),d)$, where $\Var(X(T))$ is
the optimal value of \eqref{eq:utility-positive} corresponding to $d$, is called an efficient
point. The set of all efficient points,
when the parameter $d$ runs over all possible values, is called
the efficient frontier.
\end{defi}
In the current setting, the admissible controls belong to a convex cone, so the expectation of the final outcome  may not be arbitrary.
Denote by $V(d)$ the optimal value of problem \eqref{eq:utility-positive}.
Denote
\begin{align}
  \widehat{d}=\sup\big\{\BE[X(T)]:\text{$X(\cdot)$ is an admissible process}\big\}.
\end{align}
%
Taking $\pi(\cdot)\equiv 0$, we see that $X(t)=x_0e^{\int_0^t r(s)\dd s}$ is an admissible process, so $\widehat{d}\geqslant \BE[X(T)]= x_0e^{\int_0^T r(s)\dd s}$.
The following nontrivial example shows that it is possible that $\widehat{d}= x_0e^{\int_0^T r(s)\dd s}$.
\par
\begin{example}
Let $B(t)=-C(t)\chi$, where $\chi$ is any positive vector of appropriate dimension. Then for any admissible control $\pi(\cdot)\in\BC$, we have $\pi(\cdot)'B(\cdot)=-\pi(\cdot)'C(\cdot)\chi\leqslant 0$. Therefore, by \eqref{eq:state-positive},
$$\dd\;( \BE[X(t)])= (r(t) \BE[X(t)]+\BE[\pi(t)'B(t)])\dd t\leqslant  r(t) \BE[X(t)] \dd t,$$
which implies $\BE[X(T)]\leqslant x_0e^{\int_0^T r(s)\dd s}$. Hence $\widehat{d} = x_0e^{\int_0^T r(s)\dd s}$.
\end{example}
\begin{thm}
Assume that $\widehat{d} = x_0e^{\int_0^T r(s)\dd s}$. Then the optimal value of problem \eqref{eq:utility-positive} is 0.
\end{thm}
\begin{proof}
From Assumption \ref{ass:d} and with $\widehat{d} = x_0e^{\int_0^T r(s)\dd s}$, we obtain that the only possible value of $d$ is  $x_0e^{\int_0^T r(s)\dd s}$.
Note that  $(X(t), \pi(t))\equiv (x_0e^{\int_0^t r(s)\dd s}, 0)$ is an admissible pair satisfying the constraint of problem \eqref{eq:utility-positive},
so $V(d)\leqslant \BE[X(T)-d]^2=\BE[x_0e^{\int_0^T r(s)\dd s}-d]^2=0$. The proof is complete.
\end{proof}

From now on we assume $\widehat{d}> x_0e^{\int_0^T r(s)\dd s}$.
Denote $\setd=(0,\widehat{d}~)$ and $\setdp=\big[x_0e^{\int_0^T r(s)\dd s},\widehat{d}~\big)$.
In this case both $\setd$ and $\setdp$ are nonempty intervals.


\begin{lemma}\label{value:vconvex}
The value function $V(\cdot)$ is convex on $\setd$ and strictly  increasing on $\setdp$.
\end{lemma}
\begin{proof}
Let  $ (\overline{X}(\cdot), \overline{\pi}(\cdot ))$ and   $(\widetilde{X}(\cdot), \widetilde{\pi}(\cdot ))$
be any two admissible pairs such that $d_1=\BE[\overline{X}(T)]$  and $d_2=\BE[ \widetilde{X}(T)] $ are different and both in $\setd$.
For any $0<\alpha<1$, define $ ( \widehat{X}(\cdot), \widehat{\pi}(\cdot )) = \big(\alpha\overline{X}(\cdot)+(1-\alpha)\widetilde{X}(\cdot), \alpha\overline{\pi}(\cdot )+ (1-\alpha)\widetilde{\pi}(\cdot )\big)$. Then  $ ( \widehat{X}(\cdot), \widehat{\pi}(\cdot )) $ satisfies \eqref{eq:state-positive} such that $ \widehat{\pi}(\cdot )\in\BC$, $\widehat{X}(\cdot)\geqslant  0$ and $\BE[\widehat{X}(T)] =\alpha d_1+(1-\alpha)d_2\in\setd$, that is, $ ( \widehat{X}(\cdot), \widehat{\pi}(\cdot )) $ is an admissible pair.
 Therefore,
\begin{multline*}
V\left(\alpha d_1+(1-\alpha)d_2\right)\leqslant   \Var( \widehat{X}(T))
 =\Var(\alpha\overline{X}(\cdot)+(1-\alpha)\widetilde{X}(\cdot))\\
 \leqslant   \alpha \Var(\overline{X}(T))+(1-\alpha)\Var(\widetilde{X}(T)),
\end{multline*}
where we used the convexity of square function to get the last inequality.
Because $(\overline{X}(\cdot), \overline{\pi}(\cdot ))$ and $(\widetilde{X}(\cdot), \widetilde{\pi}(\cdot ))$ are arbitrary chosen, we conclude that
\begin{align*}
V\left(\alpha d_1+(1-\alpha)d_2\right)\leqslant    \alpha V(d_1)+ (1-\alpha)V(d_2).
\end{align*}
This establishes the  convexity of $V(\cdot)$.
\par
Taking $\pi(\cdot)\equiv 0$, we see that $X(T)=x_0e^{\int_0^T r(s)\dd s}$. This clearly implies that
\begin{align*}
V\left( x_0e^{\int_0^T r(s)\dd s}\right) =0.
\end{align*}
It is known  that if there are no portfolio constraints (i.e. $C(t)=0$ for all $t\in[0,T]$), then the optimal value is positive on  $  \setdp$ (see \cite{BJPZ}),  so $V(\cdot)$ must be positive on  $\setdp$. Using the convexity, we conclude that $V(\cdot)$ is strictly increasing on  $\setdp$.
\end{proof}
\begin{coro}
The value function $V(\cdot)$ is finite and continuous on $\setd$.
\end{coro}
\par
Since problem \eqref{eq:utility-positive} is a convex optimization problem,
the mean constraint $\BE[X(T)] = d$ can be dealt with by introducing a Lagrange multiplier.  As well-known, then mean-variance portfolio selection problem \eqref{eq:utility-positive} is meaningful only when $d\in\setdp$. We will focus on this case.
\par
Because $V(\cdot)$ is convex on $\setd$ and strictly increasing at any $d \in \big(x_0e^{\int_0^T r(s)\dd s},\widehat{d}~\big)$, there is a constant $\lambda>  0$ such that $V(x)-2\lambda  x\geqslant  V(d)-2\lambda d$ for all $x\in \setd$, where the factor ${2}$ in front of the multiplier $\lambda $ is introduced   just for convenience.
In this way the portfolio selection problem \eqref{eq:utility-positive} is equivalent to the following problem
\begin{eqnarray}\label{eq:Pmu}
\begin{array}{ll}
       \min &\quad \BE[X(T) - d]^2 -2\lambda(\BE[X(T)] - d), \\ [2mm]
       \mathrm{subject\ to}&\quad
       \left\{\begin{array}{l}
       \pi(\cdot) \in\BC \mbox{ and } X(\cdot) \geqslant  0, \\ [2mm]
       (X(\cdot), \pi(\cdot )) ~ \hbox{\rm satisfies equation \eqref{eq:state-positive}},
       \end{array}\right.
\end{array}
\end{eqnarray}
or equivalently,
\begin{eqnarray*}
\begin{array}{ll}
       \min &\quad \BE[X(T) - (d+\lambda)]^2, \\ [2mm]
       \mathrm{subject\ to}&\quad
       \left\{\begin{array}{l}
       \pi(\cdot) \in\BC \mbox{ and } X(\cdot) \geqslant  0, \\ [2mm]
       (X(\cdot), \pi(\cdot )) ~ \hbox{\rm satisfies equation \eqref{eq:state-positive}}
       \end{array}\right.
\end{array}
\end{eqnarray*}
in the sense that these problems have exactly the same optimal pair if one of them admits one.
\par
We plan to use dynamic programming to study the aforementioned problems, so we
 denote by $\hat V(t,x) $   the optimal value of problem
 \begin{eqnarray}\label{eq:Pmu-1}
\begin{array}{ll}
       \min &\quad \BE[X(T) - (d+\lambda)|\BF_t, X(t)=x]^2, \\ [2mm]
       \mathrm{subject\ to}&\quad
       \left\{\begin{array}{l}
       \pi(\cdot) \in\BC \mbox{ and } X(\cdot) \geqslant  0, \\ [2mm]
       (X(\cdot), \pi(\cdot )) ~ \hbox{\rm satisfies equation \eqref{eq:state-positive}}
       \end{array}\right.
\end{array}
\end{eqnarray}
\begin{lemma}\label{v-convexincreasing}
The  function $\widehat{V}(t, \cdot)$ is strictly decreasing and  convex  on $\big(0, (d+\lambda)e^{-\int_t^T r(s)\dd s}\big]$ for each fixed $t\in[ 0,T]$.
\end{lemma}
\begin{proof}
The proof is similar to that of  Lemma \ref{value:vconvex}. We leave the proof for the interested readers.
\end{proof}

\begin{remark}
If the initial wealth  $X(t)=x$ is too big, then as well-known the mean-variance portfolio selection problem \eqref{eq:Pmu-1} is not meaningful. This make us focus on the small initials in $\big(0, (d+\lambda)e^{-\int_t^T r(s)\dd s}\big]$.
\end{remark}

\begin{lemma}\label{X=0}
If $X(\cdot)$ is a feasible wealth process with $X(t)=0$ for some $t\in[0,T]$, then
$X(s)=0$ for all $s\in[t,T]$.
\end{lemma}
\begin{proof}
Since $X(\cdot)$ is a feasible wealth process, we have $X(s)\geqslant 0$, for all $s\in[t,T]$. If $\BP(X(s)>0)$ is positive for some $s\in[t,T]$, then this would lead to an arbitrage opportunity.
\end{proof}

\begin{lemma}\label{v-boundary values}
We have that $\widehat{V}(t,0)=(d+\lambda)^2$ and $\widehat{V}\left(t, (d+\lambda)e^{-\int_t^T r(s)\dd s}\right) = 0$  for all $t\in[0,T]$.
\end{lemma}
\begin{proof}
If $X(t)=0$, then $X(T)=0$ by Lemma \ref{X=0}. Hence, $\widehat{V}(t,0)=(d+\lambda)^2$.
\par
Suppose $X(t)=(d+\lambda)e^{-\int_t^T r(s)\dd s}$. Then taking $\pi(\cdot)\equiv 0$, we obtain that $X(T)=d+\lambda$, so
$\widehat{V}\left(t, (d+\lambda)e^{-\int_t^T r(s)\dd s}\right) \leqslant \BE[X(T)-(d+\lambda)]^2=0$. The proof is complete.
\end{proof}


\section{An Equivalent Stochastic Problem}
\noindent
Since the Riccati equation approach to solve problem \eqref{eq:Pmu-1}  is not applicable in this case,
we consider the corresponding Hamilton-Jacobi-Bellman (HJB) equation. This
is the following partial differential equation:
\begin{equation}\label{eq:HJB-positive}
\begin{cases}
\mathcal{L}v= 0, \hfill \quad (t,x)\in\sets,\\
v\left(t, (d+\lambda)e^{-\int_t^T r(s)\dd s}\right) = 0, \quad v(t,0)=(d+\lambda)^2,\quad\qquad\hfill 0 \leqslant  t \leqslant  T, \\
v(T, x) =  \big(x-(d+\lambda)\big)^2, \hfill 0<x<d+\lambda,
\end{cases}
\end{equation}
where
\begin{align*}
\mathcal{L}v=&v_t(t, x) + \inf\limits_{\pi \in \BC_t}\Big\{v_x(t, x)[r(t)x + \pi'B(t)]
+ \frac{1}{2}v_{xx}(t, x)\pi'\sigma(t)\sigma(t)'\pi\Big\},\\
\sets=&\left\{(t,x): 0 \leqslant  t < T, \;0<x<(d+\lambda)e^{-\int_t^T r(s)\dd s}\right\},
\end{align*}
and
\begin{align*}
  \BC_t=\{z\in\R^m:C(t)'z\in \mathbb{R}_+^k\}.
\end{align*}
We need the following technical result.

\begin{lemma}\label{v:upperbound}
Suppose problem \eqref{eq:HJB-positive} admits a solution $v\in C^{1,2}(\sets)$ which is convex in the second argument. Then $v\leqslant (d+\lambda)^2$ on $\sets$.
\end{lemma}
\begin{proof}
By the convexity of $v$ in the second argument, we have, for each $(t,x)\in\sets$,
\begin{align*}
v(t,x)\leqslant \max\left\{v(t,0),v\left(t, (d+\lambda)e^{-\int_t^T r(s)\dd s}\right) \right\}=(d+\lambda)^2.
\end{align*}
The proof is complete.
\end{proof}

Now we are ready to establish the following result:
\begin{thm}
Suppose problem \eqref{eq:HJB-positive} admits a solution $v\in C^{1,2}(\sets)$ which is convex in the second argument. Then $\widehat{V}= v$ on $\sets$.
\end{thm}
\begin{proof}
Without loss of generality, we shall show $\widehat{V}(0,x_0)= v(0,x_0)$.
\par
Let  $(X(\cdot),\pi(\cdot))$ be the corresponding optimal pair. Define
\begin{align}
\tau= & \,\inf\left\{t\in[0,T]: X(t)=0\text{ or }X(t)=(d+\lambda)e^{-\int_t^T r(s)\dd s}\right\}\wedge T,\\
\tau_N=& \,\sup\left\{t\in[0,T]:\int_0^t\Vert v_x(s,X(s))\pi(s)'\sigma(s)\Vert^2\dd s\leqslant N\right\}\wedge T.
\end{align}
Applying It\^{o}'s Lemma to $v(t,X_t)$ yields
\begin{align*}
& v(\tau\wedge\tau_N ,X(\tau\wedge\tau_N)) \\
=& \int_0^{\tau\wedge\tau_N}\!\!\left(\!\!v_t(t, X(t)) +v_x(t, X(t))[r(t)X(t) + \pi(t)'B(t)]
+ \frac{1}{2}v_{xx}(t, X(t))\pi(t)'\sigma(t)\sigma(t)'\pi(t)\!\!\right)\!\dd t \\
&+\int_0^{\tau\wedge\tau_N}v_x(t, X(t)) \pi(t)'\sigma(t)\dd W(t)+v(0,x_0) \\
\geqslant& \int_0^{\tau\wedge\tau_N} \mathcal{L}v(t,X(t))\dd t +\int_0^{\tau\wedge\tau_N}v_x(t, X(t)) \pi(t)'\sigma(t)\dd W(t)+v(0,x_0) \\
= &\int_0^{\tau\wedge\tau_N}v_x(t, X(t))\pi(t)'\sigma(t)\dd W(t)+v(0,x_0).
\end{align*}
Taking expectation of both sides, we obtain
\begin{align*}
\BE[ v(\tau\wedge\tau_N ,X(\tau\wedge\tau_N))] \geqslant\BE\left[\int_0^{\tau\wedge\tau_N}v_x(t,X(t)) \pi(t)'\sigma(t)\dd W(t)+v(0,x_0)\right]=v(0,x_0).
\end{align*}
Because $v$ is continuous, and $\tau\wedge\tau_N$ and $X(\tau\wedge\tau_N)$ are both uniformly bounded,
letting $N\to\infty$ and applying the dominated convergence theorem, we obtain
\begin{align}\label{v-ineq0}
\BE[ v(\tau  ,X(\tau))] \geqslant v(0,x_0).
\end{align}
If $X(\tau)=0$, then by  Lemma \ref{X=0}, we have $X(T)=0$. Using Lemma \ref{v:upperbound} yields
\begin{align}\label{v-ineq1}
  \BE[v(T,X(T))|\mathcal{F}_{\tau}]\mathbf{1}_{\{X(\tau)=0\}}=(d+\lambda)^2\mathbf{1}_{\{X(\tau)=0\}}\geqslant v(\tau,X(\tau))\mathbf{1}_{\{X(\tau)=0\}}.
\end{align}
If $X(\tau)= (d+\lambda)e^{-\int_\tau^T r(s)\dd s}$, then  $v(\tau  ,X(\tau))=0$.  This trivially leads to
\begin{align}\label{v-ineq2}
  \BE[v(T,X(T))|\mathcal{F}_{\tau}]\mathbf{1}_{\{X(\tau)=(d+\lambda)e^{-\int_\tau^T r(s)\dd s}\}}\geqslant v(\tau,X(\tau))\mathbf{1}_{\{X(\tau)=(d+\lambda)e^{-\int_\tau^T r(s)\dd s}\}}.
\end{align}
If $0<X(\tau) < (d+\lambda)e^{-\int_\tau^T r(s)\dd s}$, then $\tau=T$ by the definition. Hence,
\begin{multline} \label{v-ineq3}
  \BE[v(T,X(T))|\mathcal{F}_{\tau}]\mathbf{1}_{\{0<X(\tau)<(d+\lambda)e^{-\int_\tau^T r(s)\dd s}\}}
 =\BE[v(\tau,X(\tau))|\mathcal{F}_{\tau}]\mathbf{1}_{\{0<X(\tau)<(d+\lambda)e^{-\int_\tau^T r(s)\dd s}\}} \\
=v(\tau,X(\tau))\mathbf{1}_{\{0<X(\tau)<(d+\lambda)e^{-\int_\tau^T r(s)\dd s}\}}. \quad
\end{multline}
From \eqref{v-ineq1}, \eqref{v-ineq2} and \eqref{v-ineq3}, we obtain
\begin{gather*}
  \BE[v(T,X(T)) |\mathcal{F}_{\tau}]\geqslant  v(\tau  ,X(\tau)), \\
 \BE[v(T,X(T))]\geqslant   \BE[v(\tau  ,X(\tau))]\geqslant v(0,x_0),
\end{gather*}
where the last inequality holds due to \eqref{v-ineq0}. Note $v(T,X(T))=(X(T)-(d+\lambda))^2$, so $\widehat{V}(0,x_0)\geqslant v(0,x_0)$.
\par
On the other hand,   define a feasible portfolio in the feedback form,
\begin{equation*}
\pi(t)=\begin{cases}
0, &  \text{ if $X(t)=0$;}\\
0, &  \text{ if $X(t)=(d+\lambda)e^{-\int_t^T r(s)\dd s}$;}\\
\argmin\limits_{\pi \in \BC_t}\Big\{ v_x(t, X(t)) \pi'B(t)+ \frac{1}{2}v_{xx}(t, X(t))\pi'\sigma(t)\sigma(t)'\pi\Big\},  &  \text{ otherwise.}\\
\end{cases}
\end{equation*}
It is not hard to see $(X(\cdot),\pi(\cdot))$ is an admissible pair.
Then we see  that \eqref{v-ineq0},  \eqref{v-ineq1}, \eqref{v-ineq2} and \eqref{v-ineq3} become identities, so
\begin{align*}
\BE[X(T)-(d+\lambda)]^2= \BE[v(T,X(T))]= \BE[v(\tau  ,X(\tau))]= v(0,x_0).
\end{align*}
This implies that $\widehat{V}(0,x_0)\leqslant v(0,x_0)$. The proof is complete.
\end{proof}

Before we go further, we need the following key result.
\begin{lemma}\label{lem:key}
Suppose $\BA \in \mathbb{R}^{m\times k}$, $\BB \in \mathbb{R}^m$,  $\BD \in \mathbb{R}^{m\times m}$
and $\BD\BD' $ is positive definite, and $\BCC=\{z\in\R^m: \BA'z\in \R^k_{+}\}$.
Then, for each fixed scalar $\alpha > 0$, the following two convex optimization problems
\begin{align} \label{min:c_t}
 \min_{z \in \BCC} \tfrac{1}{2}z'\BD\BD'z - \alpha \BB'z
\end{align}
and
\begin{align} \label{min:c_t-mu_t}
\min_{z \in \R^m}  \tfrac{1}{2}z'\BD\BD'z - \alpha \bar{z}'\BD\BD'z
\end{align}
have the same optimal solution $\alpha \bar{z}$ and the same optimal value $ -\frac{1}{2}\alpha^2  \bar{z}\BD'\BD\bar{z} $, where
\begin{align}\label{eq:s1}
\bar{z}=\argmin_{z \in \BCC}  \Vert \BD'z - \BD^{-1}\BB\Vert.
\end{align}
\end{lemma}
\begin{proof}
Because $\BCC$ is a cone, it is sufficient to study the case $\alpha=1$.
From the definition, $\bar{z}$ solves
\begin{align}\label{eq:min-cone}
 \min_{z \in\BCC} \tfrac{1}{2}z'\BD\BD'z - \BB'z.
\end{align}
By the definition of  $\BCC$, the above problem is equivalent to
\begin{align*}
 \min_{z\in\R^m, \BA'\!z\in \R^k_{+}} \tfrac{1}{2}z'\BD\BD'z - \BB'z.
\end{align*}
Introducing a Lagrange multiplier yields the following unconstrained problem
\begin{align}\label{eq:min-unconstrint}
 \min_{z \in \R^m} \tfrac{1}{2}z'\BD\BD'z -   \BB'z-\nu'\BA'z.
\end{align}
Problems (\ref{eq:min-cone}) and (\ref{eq:min-unconstrint}) should have the same unique solution $\bar{z}$ and optimal value, for some $\nu\in\R^{k}$.
Since $(\BD\BD')^{-1}(\BB+\BA\nu)$  is the unique solution to problem (\ref{eq:min-unconstrint}),
we have $\bar{z}=(\BD\BD')^{-1}(\BB+\BA\nu)$ and $\BD\BD'\bar{z}=\BB+\BA\nu$.
Therefore, the optimal solution and the optimal value of the above problem (\ref{eq:min-unconstrint})
are the same as those of the following unconstrained problem
\begin{align}
 \min_{z \in \R^m} \tfrac{1}{2}z'\BD\BD'z -  \bar{z}' \BD\BD'z.
\end{align}
Note that $\alpha=1$, so the above problem is actually the same as problem \eqref{min:c_t-mu_t}.
The proof is complete.
\end{proof}

\begin{remark}
By Lemma \ref{v-convexincreasing}, we know that the solution $\widehat{V}(t, \cdot)$ to problem (\ref{eq:HJB-positive})
is strictly decreasing and convex on $\big(0, (d+\lambda)e^{-\int_t^T r(s)\dd s}\big]$ for $t\in[ 0,T]$.
Therefore, $v_x(t,x) < 0$ and $v_{xx}(t,x) > 0$ for $t\in[ 0,T]$.
\end{remark}

\par
We now return to the dynamic problem \eqref{eq:HJB-positive}. Let
\begin{align}\label{eq:bar-z}
\bar{z}(t)&:=\argmin_{z  \in \BC_t}  \Vert\sigma(t)'z - \sigma(t)^{-1}B(t)\Vert.
\end{align}
By Lemma \ref{lem:key} with $\alpha = -\frac{v_x(t,x)}{v_{xx}(t,x)}>0$, the infimum in the HJB equation \eqref{eq:HJB-positive} is attained by
\begin{align}\label{eq:pi-v}
\pi = -\frac{v_x(t,x)}{v_{xx}(t,x)}\bar z(t) \in\BC_t.
\end{align}
Moreover,  the HJB equation \eqref{eq:HJB-positive}  is equivalent to
\begin{align}\label{eq:HJB-positive1}
\begin{cases}
v_t(t, x) + \inf\limits_{\pi \in \R^m}\Big\{v_x(t, x)[r(t)x + \pi' \widehat{B}(t)]
+ \frac{1}{2}v_{xx}(t, x)\pi'\sigma(t)\sigma(t)'\pi\Big\} = 0,\qquad \\
\hfill (t,x)\in\sets,\\
v\left(t, (d+\lambda)e^{-\int_t^T r(s)\dd s}\right) = 0, \quad v(t,0)=(d+\lambda)^2,\hfill 0 \leqslant  t \leqslant  T, \\
v(T, x) =  \big(x-(d+\lambda)\big)^2, \hfill 0<x<d+\lambda,
\end{cases}
\end{align}
where
\begin{align*}
   \widehat{B}(t) &:=\sigma(t)\sigma(t)'\bar{z}(t).
\end{align*}
In fact, the solution to the above HJB equation also is the value function associated with the following problem
\begin{eqnarray}\label{eq:Pmu-2}
\begin{array}{ll}
       \min &\quad \BE[X(T) - (d+\lambda)|\BF_t, X(t)=x]^2, \\ [2mm]
       \mathrm{subject\ to}&\quad
       \left\{\begin{array}{l}
         \pi(\cdot) \in  \widehat{\BC}\mbox{ and } X(\cdot) \geqslant  0, \\ [2mm]
       (X(\cdot), \pi(\cdot )) ~ \hbox{\rm satisfies the following equation \eqref{eq:state-positive2}},
       \end{array}\right.
\end{array}
\end{eqnarray}
where
\begin{align}\label{defi:C0}
\widehat{\BC}=\big\{\pi(\cdot): \pi(\cdot)\in L_{\BF}^2(0, T; \R^m)\big\},
\end{align}
and
\begin{equation}\label{eq:state-positive2}
\begin{cases}
\dd X(t) = [r(t)X(t) + \pi(t)'\widehat{B}(t)]\dd t + \pi(t)'\sigma(t)\dd W(t),  \\
 X(0) =x_0.
\end{cases}
\end{equation}
Removing the Lagrange multiplier, the above problem has the same
optimal control as the following mean-variance problem without
constraints on the portfolio:
\begin{eqnarray}\label{eq:utility-positive12}
\begin{array}{ll}
       \min & \quad\Var (X(T))= \BE[X(T) - \tilde{d}]^2, \\ [2mm]
         \mathrm{subject\ to} &\quad
       \left\{\begin{array}{l}
       \BE[X(T)] = \tilde{d}, \\ [2mm]
       \pi(\cdot) \in \widehat{\BC}\mbox{ and } X(\cdot) \geqslant  0, \\ [2mm]
       (X(\cdot), \pi(\cdot )) ~ \hbox{\rm satisfies equation \eqref{eq:state-positive2}},
       \end{array}\right.
\end{array}
\end{eqnarray}
for some $\tilde{d}$. Because the optimal solution to the above problem
\eqref{eq:utility-positive12} is also optimal to problem
\eqref{eq:utility-positive}, the mean of the optimal terminal wealth
should be the same, that is to say $\tilde{d}=d$. Therefore, we conclude that
problem  \eqref{eq:utility-positive} and problem
\begin{eqnarray}\label{eq:utility-positive2}
\begin{array}{ll}
       \min & \quad\Var (X(T))= \BE[X(T) - d]^2, \\ [2mm]
         \mathrm{subject\ to} &\quad
       \left\{\begin{array}{l}
       \BE[X(T)] = d, \\ [2mm]
       \pi(\cdot) \in  \widehat{\BC}  \mbox{ and } X(\cdot) \geqslant  0, \\ [2mm]
       (X(\cdot), \pi(\cdot )) ~ \hbox{\rm satisfies equation \eqref{eq:state-positive2}},
       \end{array}\right.
\end{array}
\end{eqnarray}
have the same optimal solution.
\par
The above mean-variance with bankruptcy prohibition problem was fully solved in \cite{BJPZ}.
Consequently, so is our problem \eqref{eq:utility-positive}. Moreover, these two problems have the same efficient frontier.
\par
It is interesting to note that 
the market price of risk $\widehat\theta(\cdot) = \sigma(\cdot)^{-1}\widehat{B}(\cdot)$
does {\it not} depend on the wealth process $X(\cdot)$.
This important feature allows us to give a linear feedback policy in $X(\cdot)$ at or before the terminal time.

\section{Optimal Portfolio}\label{sec:portfolio}
\noindent%
The result of the martingale pricing theory states that the set of random terminal payoffs
that can be generated by feasible trading strategies corresponds to
the set of nonnegative ${\cal F}_T$-measurable random payoffs $X(T)$
which satisfy the budget constraint $\BE[\phi(T)X(T)] \leqslant x_0$.
Therefore, the dynamic problem  \eqref{eq:utility-positive2}, of choosing an optimal trading strategy $\pi(\cdot)$, is equivalent to the static problem of choosing an optimal payoff $X(T)$:
\begin{align}\label{eq:MV-static}
\begin{array}{ll}
       \min & \quad\Var (X(T))= \BE[X(T) - d]^2, \\ [2mm]
       \mathrm{subject\ to} & \quad
       \left\{\begin{array}{l}
       \BE[X(T)] = d, \\ [2mm]
       \BE[\phi(T)X(T)] = x_0, \\ [2mm]
       X(T) \geqslant 0,
       \end{array}\right.
\end{array}
\end{align}
where  $\phi(\cdot)$ is the \textit{state price density}, or \textit{stochastic discount factor}, defined by
\begin{align}\label{eq:phi}
\left\{\begin{array}{rcl}
\dd\phi(t) & = & \phi(t)\{-r(t)\dd t - \widehat\theta(t)'\dd W(t)\}, \\ [2mm]
 \phi(0) & = & 1,
\end{array}\right.
\end{align}
and
\begin{align*}
\widehat\theta(t)= \sigma(t)^{-1}\widehat{B}(t)=\sigma(t)'\bar{z}(t).
\end{align*}
The above static optimization problem (\ref{eq:MV-static}) was solved in \cite{BJPZ}. The optimal random terminal payoff is
\begin{align}\label{eq:X-T-sol}
X^*(T) = (\mu-\gamma\phi(T))^+,
\end{align}
where $x^+=\max\{x,0\}$, $\mu > 0$, $\gamma > 0$ and $(\mu,\gamma)\in\R^2$ solves the system of equations
\begin{align}\label{eq:sys}
\left\{\begin{array}{l}
\BE[(\mu-\gamma\phi(T))^+] = d, \\ [2mm]
\BE[\phi(T)(\mu-\gamma\phi(T))^+] = x_0.
\end{array}\right.
\end{align}
That is,
\begin{align}\label{eq:sys-1}
\left\{\begin{array}{l}
\!\mu N\Big(\frac{\ln\big(\frac{\mu}{\gamma}\big)+\int_0^T[r(s)+\frac{1}{2}\vert\widehat\theta(s)\vert^2]\dd s}{\sqrt{\int_0^T\vert\widehat\theta(s)\vert^2ds}}\Big)
\!-\! \gamma e^{-\int_0^Tr(s)\dd s}N\Big(\frac{\ln\big(\frac{\mu}{\gamma}\big)+\int_0^T[r(s)-\frac{1}{2}\vert\widehat\theta(s)\vert^2]\dd s}{\sqrt{\int_0^T\vert\widehat\theta(s)\vert^2ds}}\Big) = d, \\ [2mm]
\!\mu N\Big(\frac{\ln\big(\frac{\mu}{\gamma}\big)+\int_0^T[r(s)-\frac{1}{2}\vert\widehat\theta(s)\vert^2]\dd s}{\sqrt{\int_0^T\vert\widehat\theta(s)\vert^2ds}}\Big)
\!-\! \gamma e^{-\int_0^T[r(s)-\vert\widehat\theta(s)\vert^2]\dd s}N\Big(\frac{\ln\big(\frac{\mu}{\gamma}\big)+\int_0^T[r(s)-\frac{3}{2}\vert\widehat\theta(s)\vert^2]\dd s}{\sqrt{\int_0^T\vert\widehat\theta(s)\vert^2ds}}\Big) = x_0e^{\int_0^Tr(s)\dd s},
\end{array}\right.
\end{align}
where $N(y) = \frac{1}{\sqrt{2\pi}}\int_{-\infty}^ye^{-\frac{t^2}{2}}\dd t$ is the cumulative distribution function of the standard normal distribution.

The investor's optimal wealth is given by the stochastic process
\begin{align}\label{eq:X-opt}
X^*(t) & =  \BE\bigg[\frac{\phi(T)}{\phi(t)}X^*(T)\bigg\vert\mathcal{F}_t\bigg] =f(t,\phi(t)),
\end{align}
where
\begin{align}
f(t,y)=  \mu N\big(-d_2(t,y)\big)e^{-\int_t^T r(s)\dd s} - \gamma N\big(-d_1(t,y)\big) y e^{-\int_t^T[2r(s)-\vert\widehat\theta(s)\vert^2]\dd s},
\end{align}
and
\begin{align*}
\begin{array}{l}
d_1(t,y) := \frac{\ln\big(\frac{\gamma}{\mu}y\big)+\int_t^T[-r(s)+\frac{3}{2}\vert\widehat\theta(s)\vert^2]\dd s}{\sqrt{\int_t^T\vert\widehat\theta(s)\vert^2ds}}, \quad
d_2(t,y) := d_1(t, y) - \sqrt{\int_t^T\vert\widehat\theta(s)\vert^2\dd s}.
\end{array}
\end{align*}
%
Applying It\^o's lemma to $f(\cdot,\phi(\cdot))$ yields
\begin{align*}
\dd X^*(t)=\dd f(t,\phi(t))=\{\cdots\}\dd t+\gamma\hat\theta(t)N\big(-d_1(t,\phi(t))\big)\phi(t)e^{-\int_t^T[2r(s)-\vert\widehat\theta(s)\vert^2]ds}\dd W(t).
\end{align*}
Comparing this to the wealth evolution equation  \eqref{eq:state-positive2},
we obtain the efficient portfolio
\begin{align}\label{eq:pi-opt-1}
\begin{array}{rcl}
\pi^*(t)
& = & \gamma(\sigma(t)\sigma(t)')^{-1}\widehat{B}(t)N\big(-d_1(t,\phi(t))\big)\phi(t)e^{-\int_t^T[2r(s)-\vert\widehat\theta(s)\vert^2]\dd s} \\ [2mm]
& = & -(\sigma(t)\sigma(t)')^{-1}\widehat{B}(t)\Big[X^*(t) - \mu N\big(-d_2(t,\phi(t))\big)e^{-\int_t^T r(s)\dd s}\Big].
\end{array}
\end{align}
\begin{remark}
The above results for the efficient portfolio and the associated wealth process were first derived in Bielecki et al. \cite{BJPZ} and give a complete solution
to the mean-variance portfolio selection problem with bankruptcy prohibition (\ref{eq:utility-positive2}).
\end{remark}
Based on the above analysis, we have the following main theorem.
\begin{thm}\label{thm:portfolio}
Assume that $\int_0^T\vert\widehat\theta(s)\vert^2\dd s > 0$. Then there exists a unique efficient portfolio for (\ref{eq:utility-positive})
corresponding to any given $d \geqslant x_0e^{\int_0^Tr(s)\dd s}$. Moreover, the efficient portfolio is given by (\ref{eq:pi-opt-1}) and
the associated wealth process is expressed by (\ref{eq:X-opt}).
\end{thm}

\section{Special Models}
\noindent
The main result of the constrained mean-variance portfolio selection model
with bankruptcy prohibition derived in Section \ref{sec:portfolio} is quite surprising but important.
The mean-variance portfolio selection model, like many other
stochastic optimization models, is
based on averaging over all the possible random scenarios.
We now discuss the model in terms of how
its different constraints could guide real investment in practice.

\subsection{Bankruptcy Prohibition with Unconstrained Portfolio}\label{sect:5.1}
\noindent
The mean-variance unconstrained portfolio problem with bankruptcy prohibition is
an interesting but practically relevant model.
In this case, $k = m$ and $\pi(\cdot) \in L_{\mathcal{F}}^2(0,T;\mathbb{R}^m)$.
It follows from (\ref{eq:bar-z}) that
\begin{align*}
\bar{z}(t) = \argmin_{z  \in \mathbb{R}^m}  \Vert\sigma(t)'z - \sigma(t)^{-1}B(t)'\Vert = (\sigma(t)\sigma(t)')^{-1}B(t)'.
\end{align*}
Therefore,
\begin{align*}
\widehat{B}(t) = \sigma(t)\sigma(t)'\bar{z}(t) = B(t).
\end{align*}

\begin{prop}\label{prop:bankruptcy}
Assume that $\int_0^T\vert\widehat\theta(s)\vert^2\dd s > 0$. Then there exists a unique efficient portfolio for this mean-variance model
corresponding to any given $d \geqslant x_0e^{\int_0^Tr(s)\dd s}$. Moreover, the efficient portfolio is given by (\ref{eq:pi-opt-1}) and
the associated wealth process is expressed by (\ref{eq:X-opt}),
where $\widehat{B}(t) = B(t)$ and $\widehat\theta(t) = \sigma(t)^{-1}B(t)$.
\end{prop}

The proof of Proposition \ref{prop:bankruptcy} can be found in Bielecki et al. \cite{BJPZ}.

\subsection{Bankruptcy Prohibition with No-shorting Constraint}\label{sect:5.2}
\noindent
The mean-variance portfolio problem with mixed bankruptcy and no-shorting constraints is
another interesting and challenging model.
In this case, $k = m$ and $\pi(\cdot) \in L_{\mathcal{F}}^2(0,T;\mathbb{R}_+^m)$.
It follows from (\ref{eq:bar-z}) that
\begin{align*}
\bar{z}(t) = \argmin_{z \in \mathbb{R}_+^m} \Vert\sigma(t)'z - \sigma(t)^{-1}B(t)'\Vert = (\sigma(t)\sigma(t)')^{-1}(B(t)+\lambda(t))',
\end{align*}
where
\begin{align}\label{eq:lambda}
\lambda(t) := \displaystyle\argmin_{y \in \mathbb{R}_+^m} \Vert\sigma(t)^{-1}y + \sigma(t)^{-1}B(t)'\Vert.
\end{align}
Therefore,
\begin{align*}
\widehat{B}(t) = \sigma(t)\sigma(t)'\bar{z}(t) = B(t)+\lambda(t).
\end{align*}

\begin{prop}\label{prop:bankruptcy-no-shorting}
Assume that $\int_0^T\vert\widehat\theta(s)\vert^2\dd s > 0$. Then there exists a unique efficient portfolio for this mean-variance model
corresponding to any given $d \geqslant x_0e^{\int_0^Tr(s)\dd s}$. Moreover, the efficient portfolio is given by (\ref{eq:pi-opt-1}) and
the associated wealth process is expressed by (\ref{eq:X-opt}),
where $\widehat{B}(t) = B(t)+\lambda(t)$ and $\widehat\theta(t) = \sigma(t)^{-1}(B(t)+\lambda(t))$.
\end{prop}

\subsection{No-shorting Constraint without Bankruptcy Prohibition}\label{sect:5.3}
\noindent
The mean-variance portfolio problem with no-shorting constraints is also
an important model in financial investment.
In this case, $k = m$ and $\pi(\cdot)\in L_{\mathcal{F}}^2(0,T;\mathbb{R}_+^m)$.
We again have
\begin{align*}
\widehat{B}(t) = B(t)+\lambda(t),
\end{align*}
where $\lambda(t)$ is determined by (\ref{eq:lambda}).

In particular, $d_1(t,\phi(t))=-\infty$ and $d_2(t,\phi(t))=-\infty$, that is,
$N\big(-d_1(t,\phi(t))\big) = N\big(-d_2(t,\phi(t))\big) = 1$.
The investor's optimal wealth is the stochastic process
\begin{align}\label{eq:X-opt-no-shorting}
\begin{array}{rcl}
X^*(t)
& = & \mu e^{-\int_t^T r(s)\dd s} - \gamma\phi(t)e^{-\int_t^T[2r(s)-\vert\widehat\theta(s)\vert^2]\dd s}
\end{array}
\end{align}
and
\begin{align}\label{eq:pi-opt-no-shorting}
\begin{array}{rcl}
\pi^*(t)
& = & -(\sigma(t)\sigma(t)')^{-1}\widehat{B}(t)\Big[X^*(t) - \mu e^{-\int_t^T r(s)\dd s}\Big],
\end{array}
\end{align}
where
\begin{align}\label{eq:pi-opt-no-shorting}
& \mu = \frac{\BE[\phi(T)^2]d-x_0\BE[\phi(T)]}{\Var(\phi(T))} = \frac{d-x_0e^{\int_0^T[r(s)-\vert\widehat\theta(s)\vert^2]\dd s}}{1-e^{-\int_0^T\vert\widehat\theta(s)\vert^2\dd s}}, \\
& \gamma = \frac{\BE[\phi(T)]d-x_0}{\Var(\phi(T))} = \frac{\big(d-x_0e^{\int_0^Tr(s)\dd s}\big)e^{\int_0^T[r(s)-\vert\widehat\theta(s)\vert^2]\dd s}}{1-e^{-\int_0^T\vert\widehat\theta(s)\vert^2\dd s}}.
\end{align}
where $\widehat\theta(t) = \sigma(t)^{-1}(B(t)+\lambda(t))$.

\begin{prop}\label{prop:no-shorting}
Assume that $\int_0^T\vert\widehat\theta(s)\vert^2\dd s > 0$. Then there exists a unique efficient portfolio for this mean-variance model
corresponding to any given $d \geqslant x_0e^{\int_0^Tr(s)\dd s}$. Moreover, the efficient portfolio is given by (\ref{eq:X-opt-no-shorting}) and
the associated wealth process is expressed by (\ref{eq:pi-opt-no-shorting}).

\end{prop}

Another version of the proof of Proposition \ref{prop:no-shorting} can be found in Li, Zhou and Lim \cite{LZL}.


\section{Example}
\noindent
In this section, a numerical example with constant coefficients is presented to
demonstrate the results in the previous section. Let $m = 3$. The
interest rate of the bond and the appreciation rate of the $m$
stocks are $r = 0.03$ and $(b_1, b_2, b_3)' = (0.12, ~0.15, ~0.18)'$,
respectively, and the volatility matrix is
$$
\sigma = \left[\begin{array}{rrr}
0.2500  &      0    &     0 \\
0.1500  & 0.2598    &     0 \\
-0.2500 & ~~~0.2887    & ~~~0.3227
\end{array}\right].
$$
Then we have
$$
\sigma^{-1} = \left[\begin{array}{rrr}
4.0000  &     0   &    0 \\
-2.3094 &  3.8490 &    0 \\
5.1640  & ~~~-3.4427 & ~~~3.0984
\end{array}\right]
$$
and
$B = (b_1 - r, b_2 - r, b_3 - r)'
= (0.09, ~0.12, ~0.15)'$.
Hence,
$$
\theta: = \sigma^{-1}B = (0.3600, ~0.2540, ~0.5164)'.
$$
In addition, we suppose that the initial prices of the stocks are $(S_1(0), S_2(0), S_3(0)) = (1, 1, 1)'$ and the initial wealth is $X(0) = 1$.

\subsection{Bankruptcy Prohibition with Unconstrained Portfolio}\label{sect:6.1}
\noindent
In this subsection, we determine the optimal portfolio and the corresponding wealth process in Subsection \ref{sect:5.1} for the above market data.
According to (\ref{eq:sys-1}), we obtain the numerical results $\mu = 1.5046$ and $\gamma = 0.3154$. Hence,
the wealth process (\ref{eq:X-opt}) can be expressed by
\begin{align}\label{eq:X-opt-1}
X^*(t) =  \mu N\big(-d_2(t,\phi(t))\big)e^{-r(T-t)} - \gamma N\big(-d_1(t,\phi(t))\big)\phi(t)e^{-[2r-\vert\hat\theta\vert^2](T-t)},
\end{align}
where
\begin{align*}
\begin{array}{ll}
d_1(t,\phi(t)) = \frac{\ln\big(\frac{\gamma}{\mu}\phi(t)\big)+[-r+\frac{3}{2}\vert\hat\theta\vert^2](T-t)}{\sqrt{\vert\hat\theta\vert^2(T-t)}}, \quad
& d_2(t,\phi(t)) = d_1(t, \phi(t)) - \sqrt{\vert\hat\theta\vert^2(T-t)}, \\ [2mm]
\phi(t) = e^{-[r+\frac{1}{2}\vert\hat\theta\vert^2](T-t)-\hat\theta(W(T)-W(t))},
& \hat\theta = \theta = \sigma^{-1}B = (0.3600, 0.2540, 0.5164)'.
\end{array}
\end{align*}
The efficient portfolio is given by
\begin{align}\label{eq:pi-opt-11}
\begin{array}{rcl}
\pi^*(t)
& = & -(\sigma\sigma')^{-1}\widehat{B}\Big[X^*(t) - \mu N\big(-d_2(t,\phi(t))\big)e^{-r(T-t)}\Big],
\end{array}
\end{align}
where
$$
(\sigma\sigma')^{-1}\widehat{B} = (\sigma\sigma')^{-1}B = (3.5200, ~-0.8000, ~1.6000)'.
$$

In particular, the policy of investing in the second stock $\pi_2^*(t)$ is negative.

\subsection{Bankruptcy Prohibition with No-shorting Constraint}\label{sect:6.2}
\noindent
From Subsection \ref{sect:6.1}, we see that there exists a shorting case in policy (\ref{eq:pi-opt-11}).
Using (\ref{eq:lambda}), we obtain the following $\lambda$ to re-construct the no-shorting policy
\begin{align}\label{eq:lambda-1}
\lambda := \displaystyle\argmin_{y \in \mathbb{R}_+^m} \Vert\sigma^{-1}y + \sigma^{-1}B'\Vert = (0, ~0.03, ~0)'.
\end{align}
Hence,
\begin{align}\label{eq:theta-B}
\left\{\begin{array}{l}
\hat\theta: = \sigma^{-1}\widehat{B} = \sigma^{-1}(B+\lambda) = (0.3600, ~0.3695, ~0.4131)', \\ [2mm]
(\sigma\sigma')^{-1}\widehat{B} = (\sigma\sigma')^{-1}(B+\lambda) = (2.72, ~0, ~1.28)'.
\end{array}\right.
\end{align}
According to (\ref{eq:sys-1}), we obtain the numerical results $\mu = 1.5253$ and $\gamma = 0.3368$. Hence,
the wealth process (\ref{eq:X-opt}) can be expressed by
\begin{align}\label{eq:X-opt-12}
X^*(t) =  \mu N\big(-d_2(t,\phi(t))\big)e^{-r(T-t)} - \gamma N\big(-d_1(t,\phi(t))\big)\phi(t)e^{-[2r-\vert\hat\theta\vert^2](T-t)},
\end{align}
where
\begin{align*}
\begin{array}{l}
d_1(t,\phi(t)) = \frac{\ln\big(\frac{\gamma}{\mu}\phi(t)\big)+[-r+\frac{3}{2}\vert\hat\theta\vert^2](T-t)}{\sqrt{\vert\hat\theta\vert^2(T-t)}}, \quad
d_2(t,\phi(t)) = d_1(t, \phi(t)) - \sqrt{\vert\hat\theta\vert^2(T-t)}, \\ [2mm]
\phi(t) = e^{-[r+\frac{1}{2}\vert\hat\theta\vert^2](T-t)-\hat\theta(W(T)-W(t))}.
\end{array}
\end{align*}
The \begin{align}\label{eq:pi-opt-12}
\begin{array}{rcl}
\pi^*(t)
& = & -(\sigma\sigma')^{-1}\widehat{B}\Big[X^*(t) - \mu N\big(-d_2(t,\phi(t))\big)e^{-r(T-t)}\Big].
\end{array}
\end{align}

Note that the policy in (\ref{eq:pi-opt-12}) is non-negative, that is, this is a no-shorting policy.

\subsection{No-shorting Constraint without Bankruptcy Prohibition}\label{sect:6.3}
\noindent
In this subsection, we present the optimal no-shorting policy without the bankruptcy prohibition of Subsection \ref{sect:5.3} and
its corresponding wealth process.
According to (\ref{eq:pi-opt-no-shorting}), we find the numerical results $\mu = 1.5095$ and $\gamma = 0.3190$. Hence,
the wealth process (\ref{eq:X-opt}) can be expressed by
\begin{align}\label{eq:X-opt-no-shorting-1}
\begin{array}{rcl}
X^*(t)
& = & \mu e^{-r(T-t)} - \gamma\phi(t)e^{-[2r-\vert\hat\theta\vert^2](T-t)}
\end{array}
\end{align}
and
\begin{align}\label{eq:pi-opt-no-shorting-1}
\begin{array}{rcl}
\pi^*(t)
& = & -(\sigma\sigma')^{-1}\widehat{B}\Big[X^*(t) - \mu e^{-r(T-t)}\Big],
\end{array}
\end{align}
where $\hat\theta$ and $(\sigma\sigma')^{-1}\widehat{B}$ are given by (\ref{eq:theta-B}), and
$$
\phi(t) = e^{-[r+\frac{1}{2}\vert\hat\theta\vert^2](T-t)-\hat\theta(W(T)-W(t))}.
$$

Note that the policy in (\ref{eq:pi-opt-no-shorting-1}) is non-negative, i.e., this is a no-shorting policy.
However, its corresponding wealth (\ref{eq:X-opt-no-shorting-1}) is possibly negative.
We shall further discuss this point by simulation results in Subsection \ref{sect:6.4}.

\subsection{Simulation}\label{sect:6.4}
\noindent
In this subsection, we further analyze using simulation how the properties
of the optimal portfolio strategies (\ref{eq:pi-opt-11}), (\ref{eq:pi-opt-12}) and (\ref{eq:pi-opt-no-shorting-1}) change according to the given target
wealth $d = 1.2X(0)$, and compare their wealth processes (\ref{eq:X-opt-1}), (\ref{eq:X-opt-12}) and (\ref{eq:X-opt-no-shorting-1}).

$$
\begin{array}{ccc}
\begin{array}{c}
\includegraphics[height=1.7in]{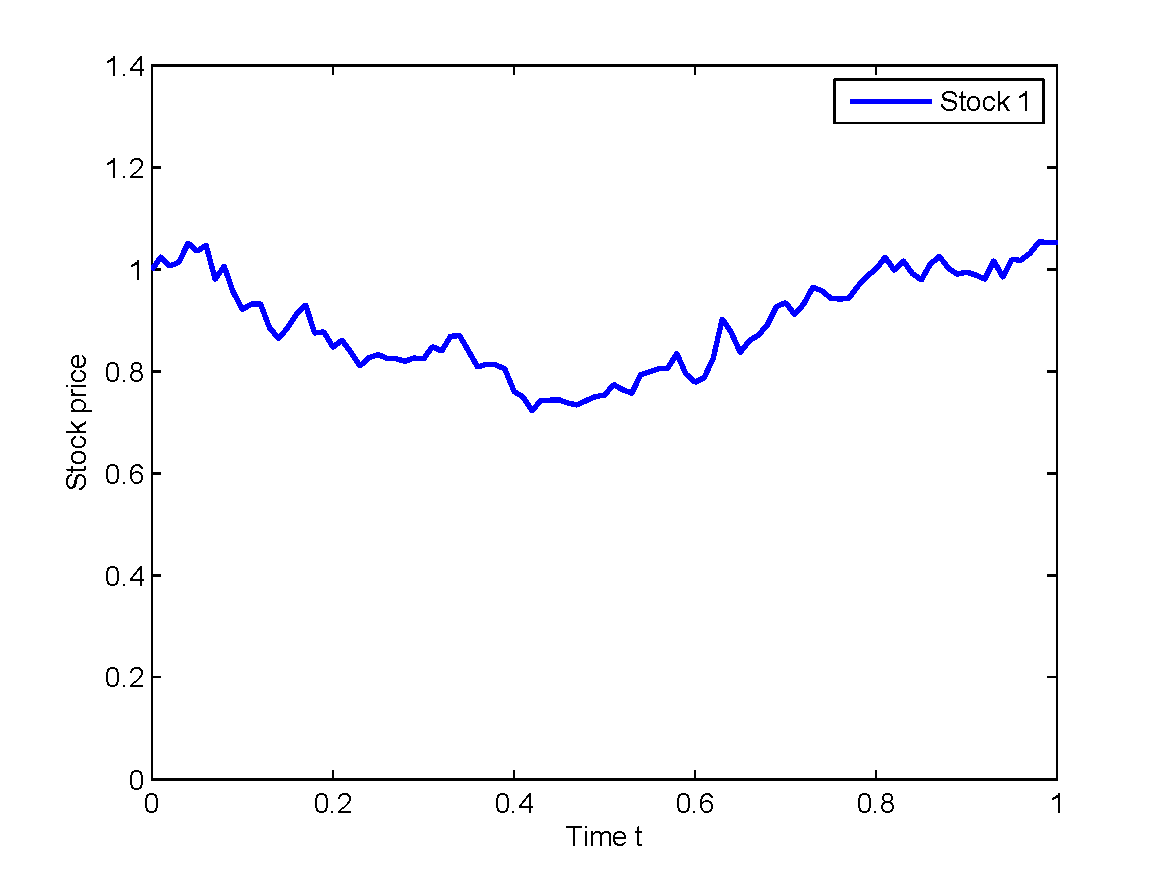}
\includegraphics[height=1.7in]{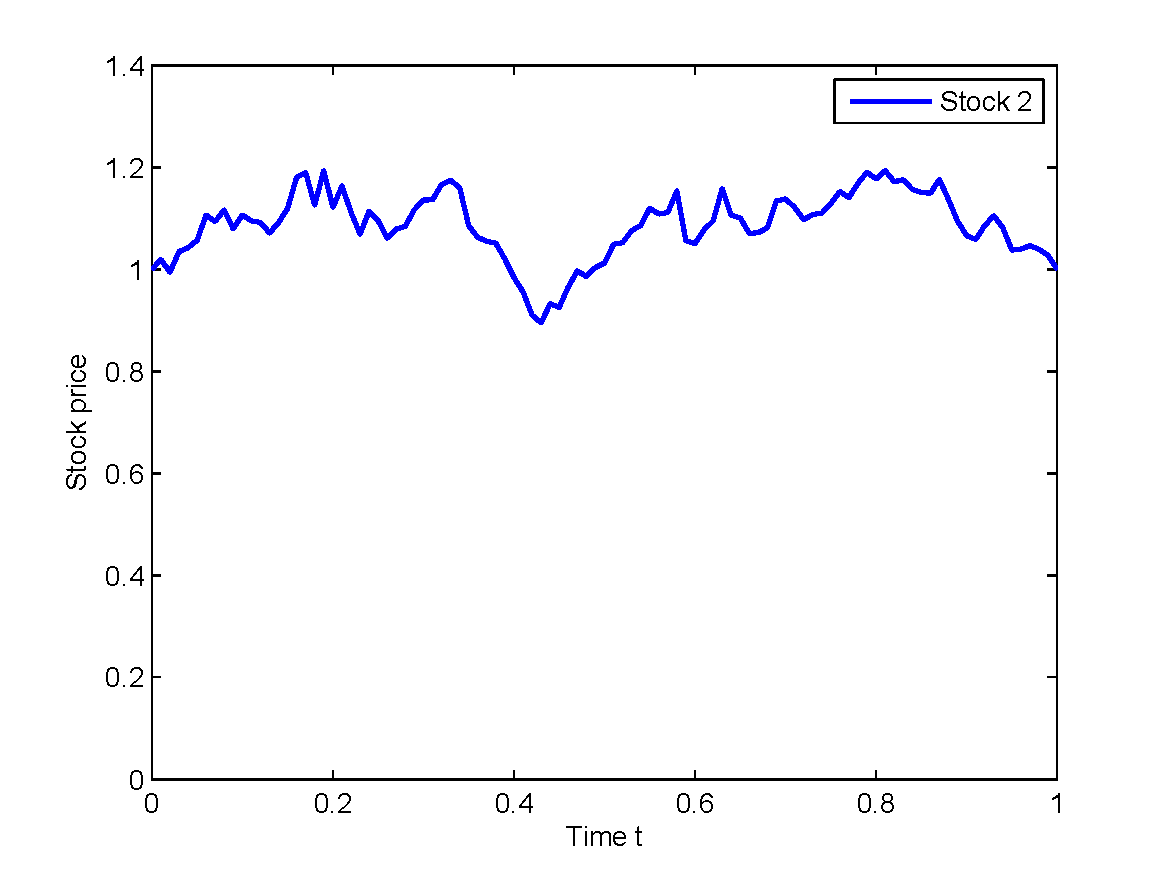}
\includegraphics[height=1.7in]{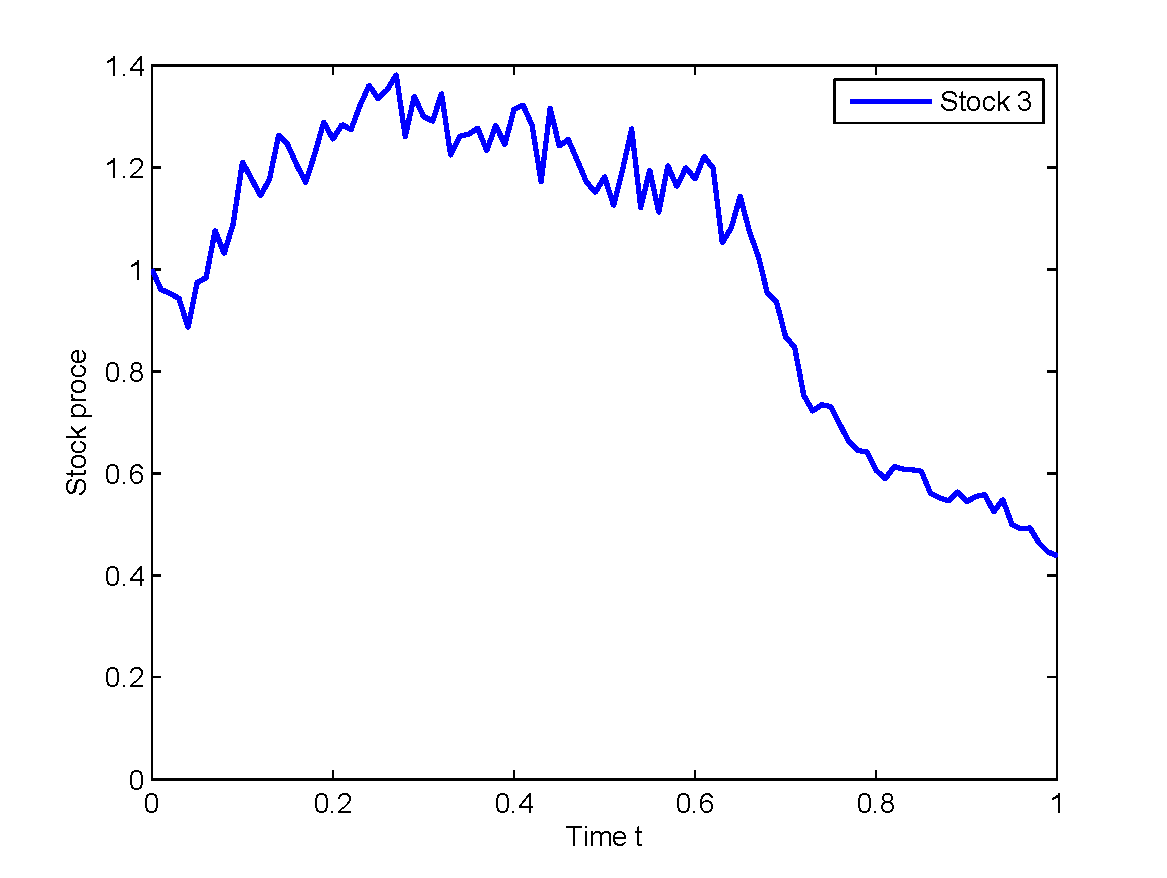} \\
\mbox{The value of boundary}
\end{array}
\end{array}
$$

$$
\begin{array}{ccc}
\begin{array}{c}
\includegraphics[height=1.7in]{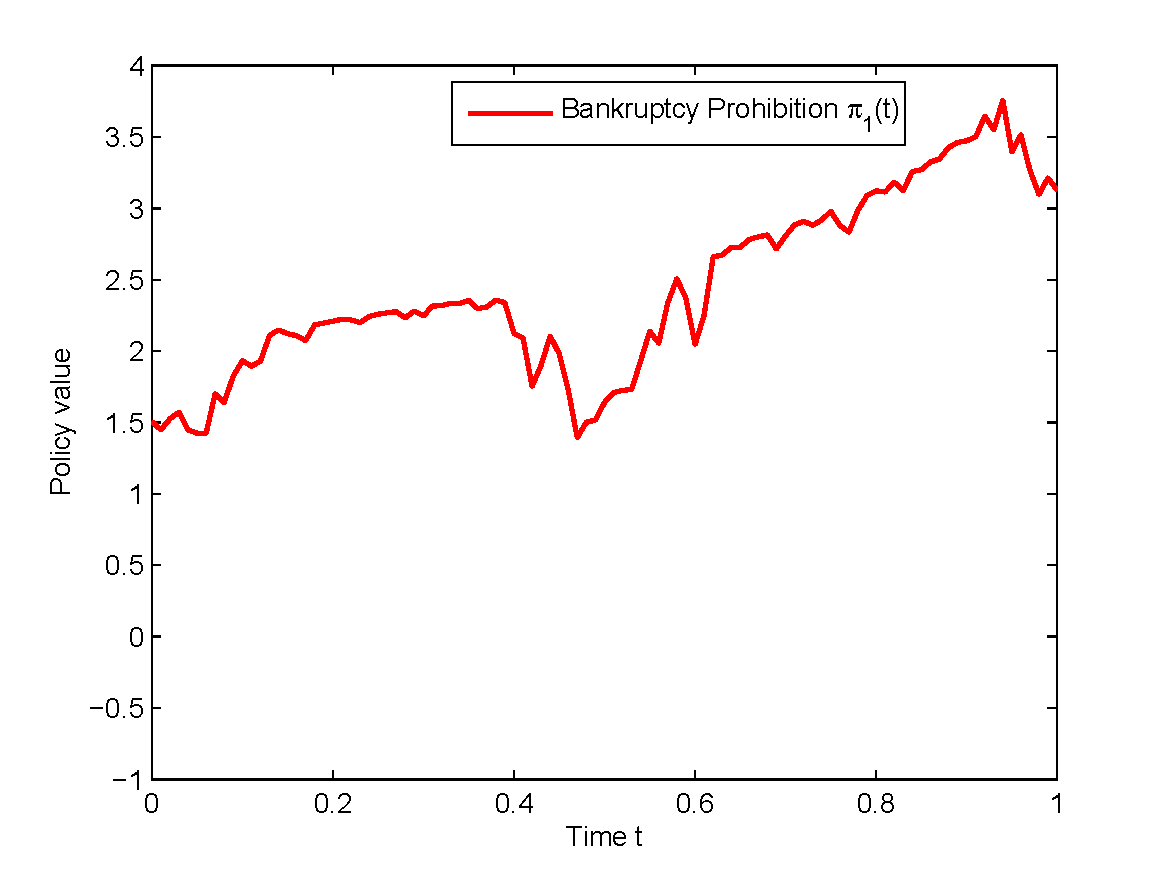}
\includegraphics[height=1.7in]{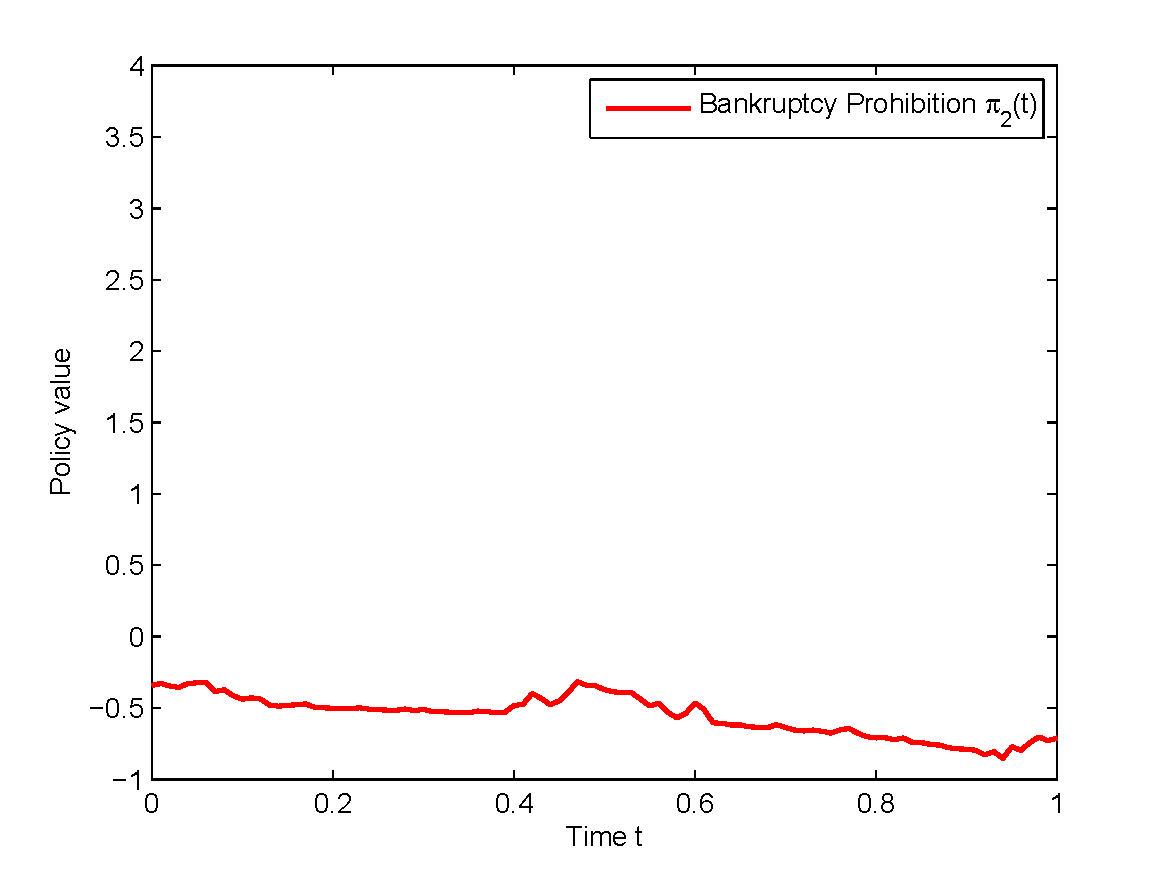}
\includegraphics[height=1.7in]{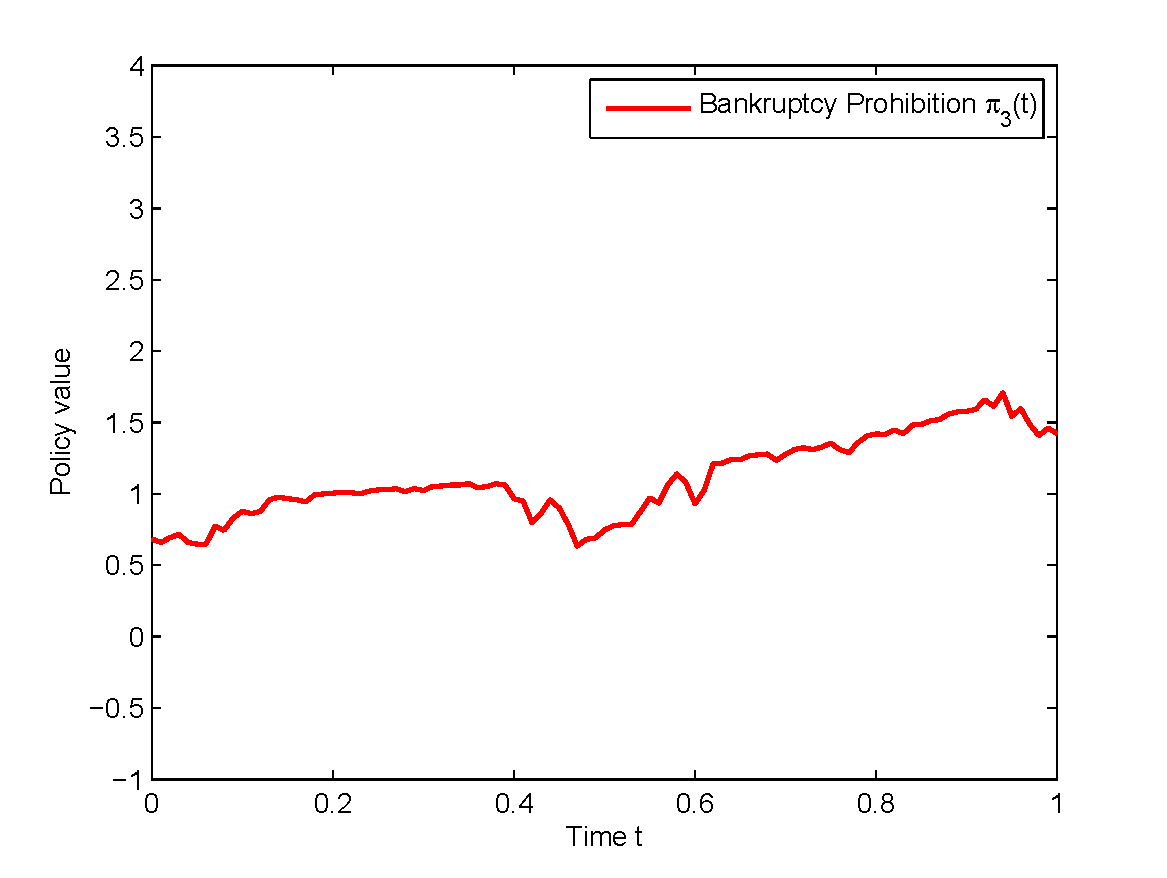} \\
\mbox{Policy of Bankruptcy Prohibition with Unconstrained Portfolio}
\end{array}
\end{array}
$$

$$
\begin{array}{ccc}
\begin{array}{c}
\includegraphics[height=1.7in]{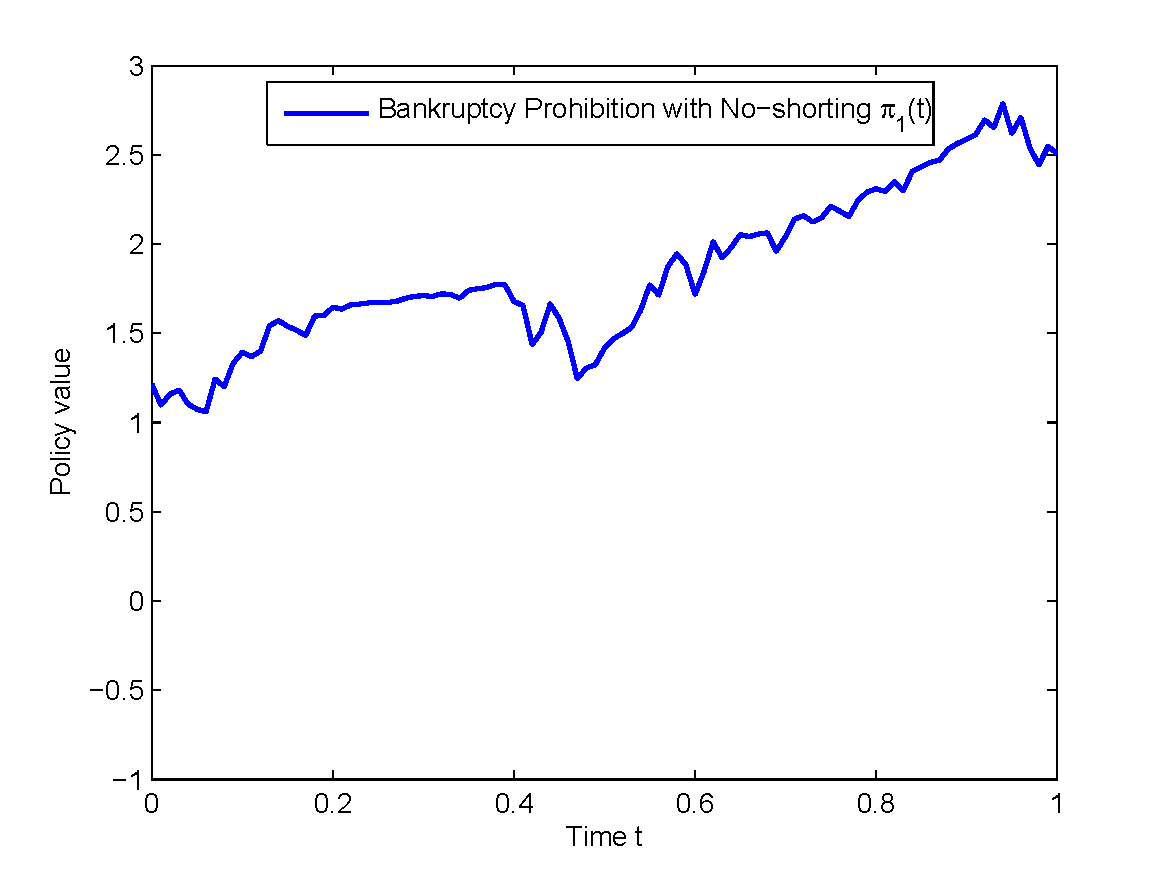}
\includegraphics[height=1.7in]{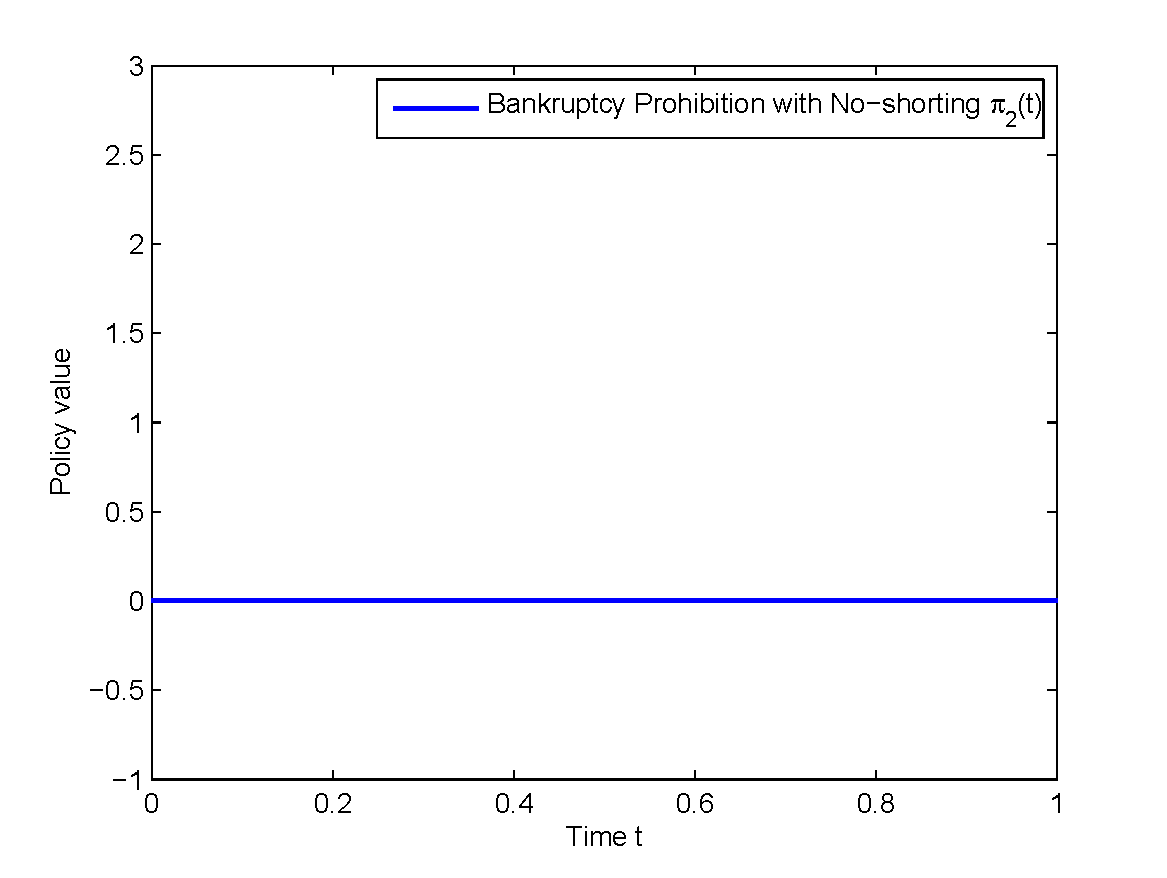}
\includegraphics[height=1.7in]{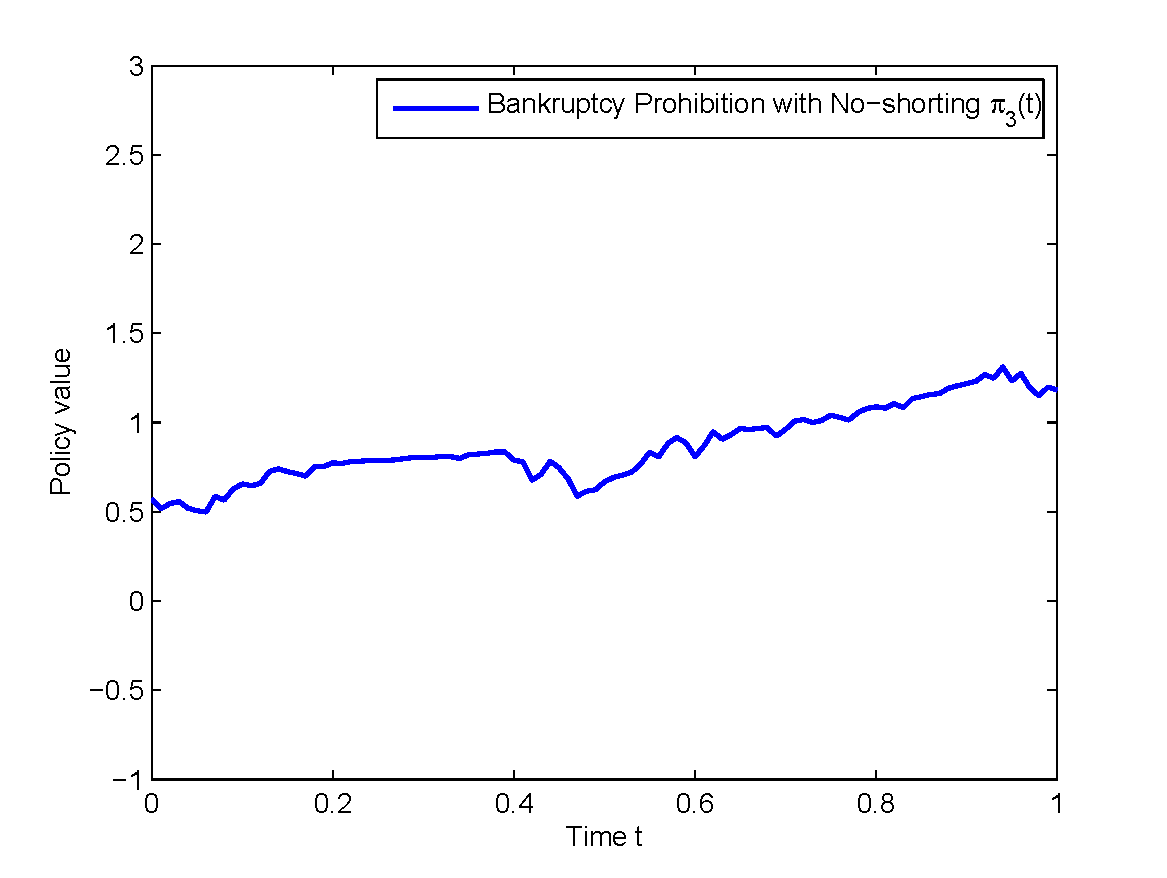} \\
\mbox{Policy of Bankruptcy Prohibition with No-shorting Constraint}
\end{array}
\end{array}
$$

$$
\begin{array}{ccc}
\begin{array}{c}
\includegraphics[height=1.7in]{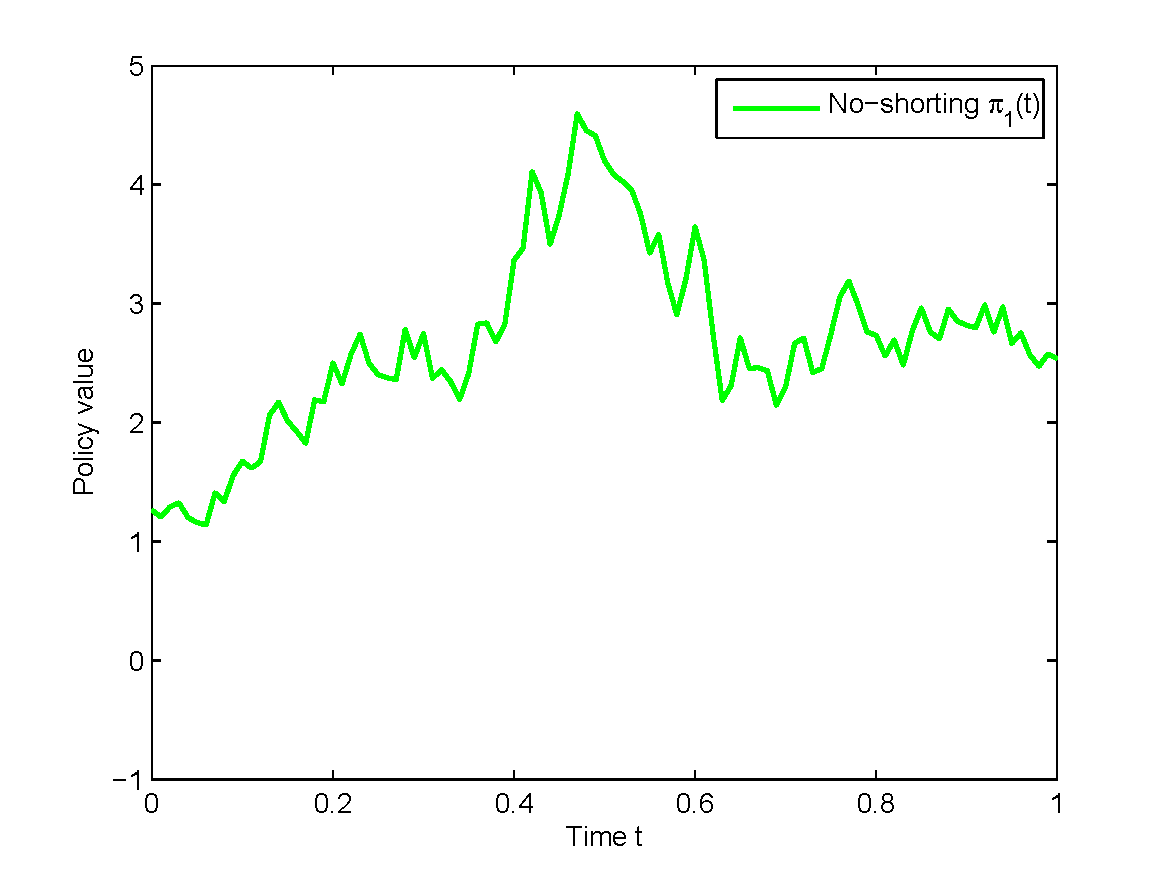}
\includegraphics[height=1.7in]{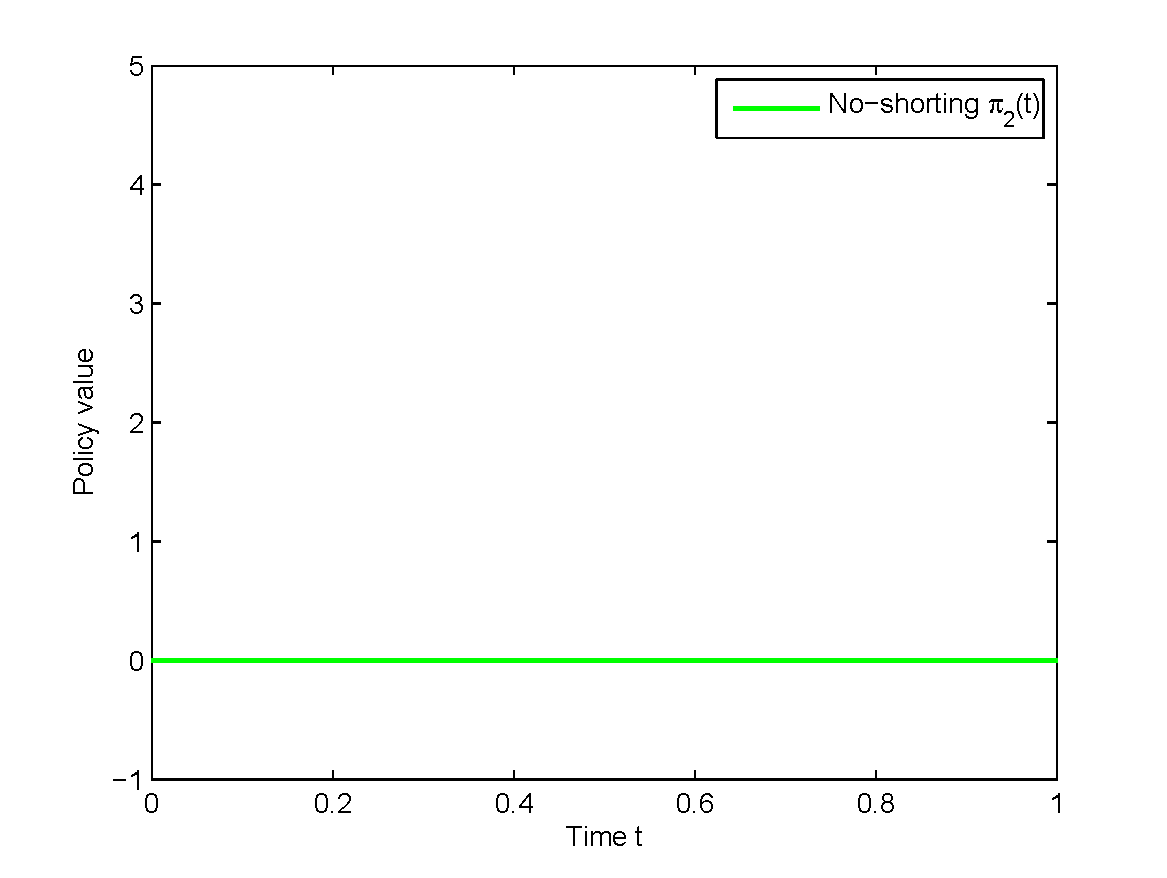}
\includegraphics[height=1.7in]{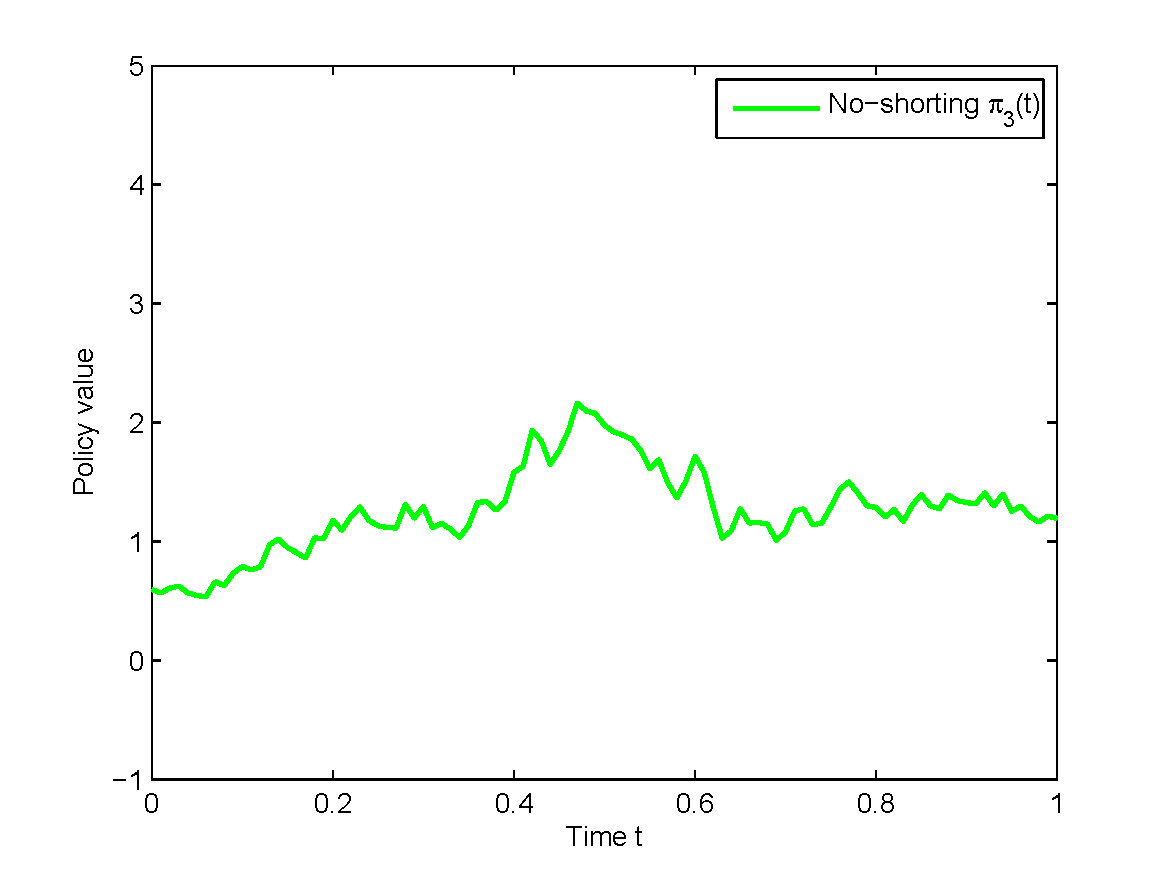} \\
\mbox{Policy of No-shorting Constraint without Bankruptcy Prohibition}
\end{array}
\end{array}
$$

$$
\begin{array}{ccc}
\begin{array}{c}
\includegraphics[height=1.7in]{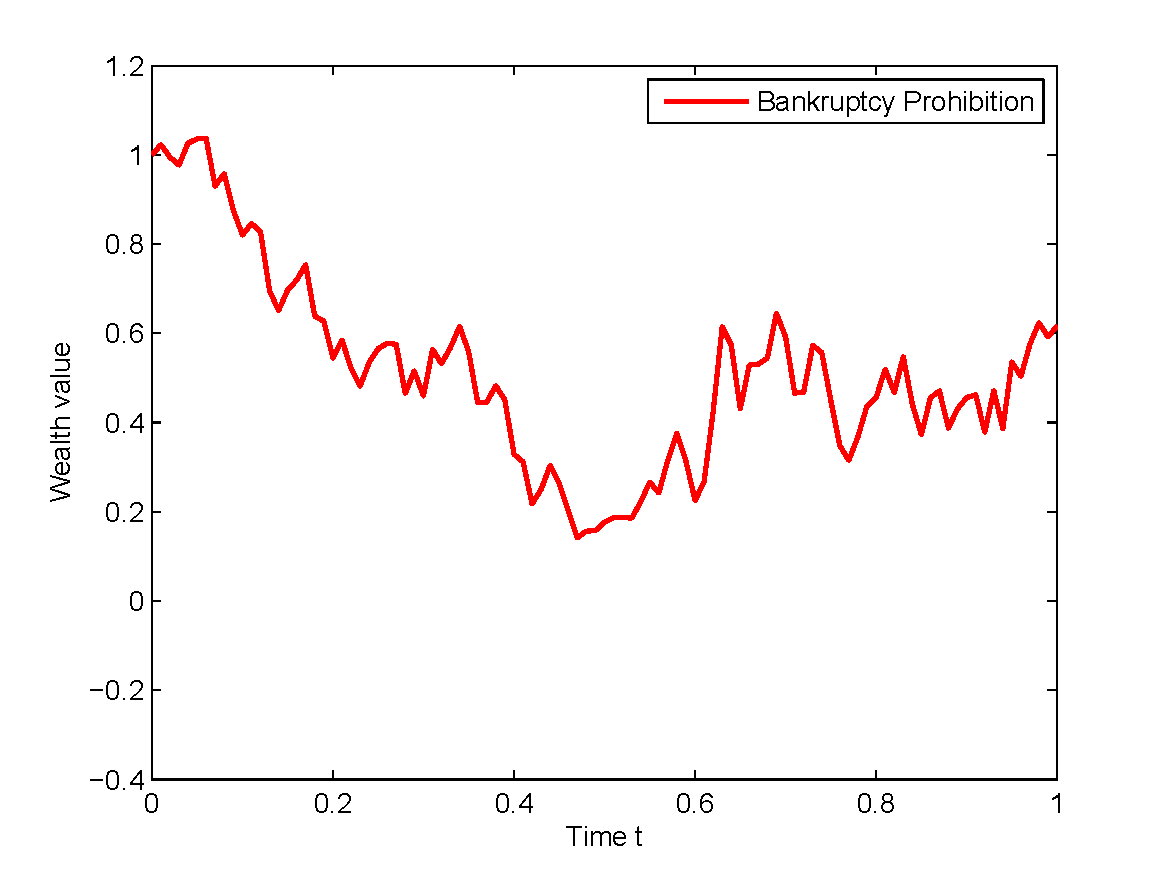}
\includegraphics[height=1.7in]{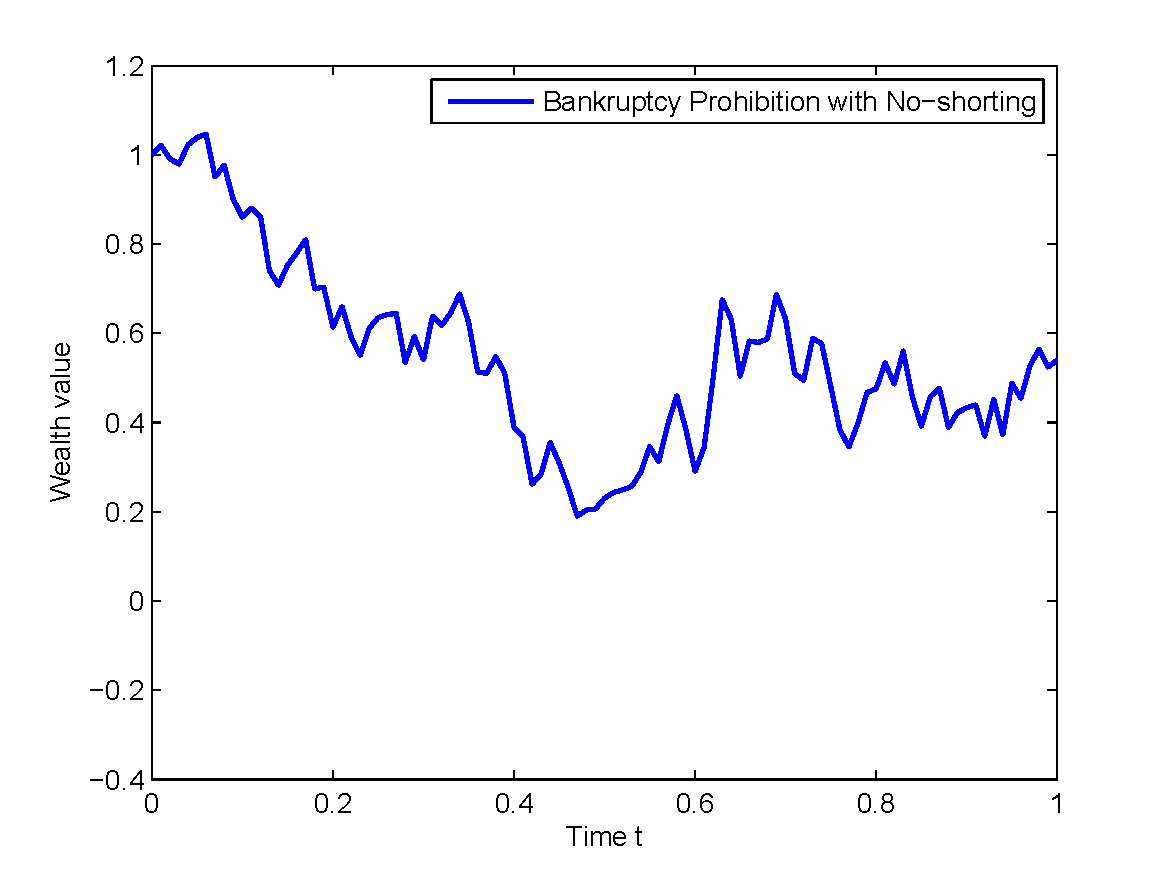}
\includegraphics[height=1.7in]{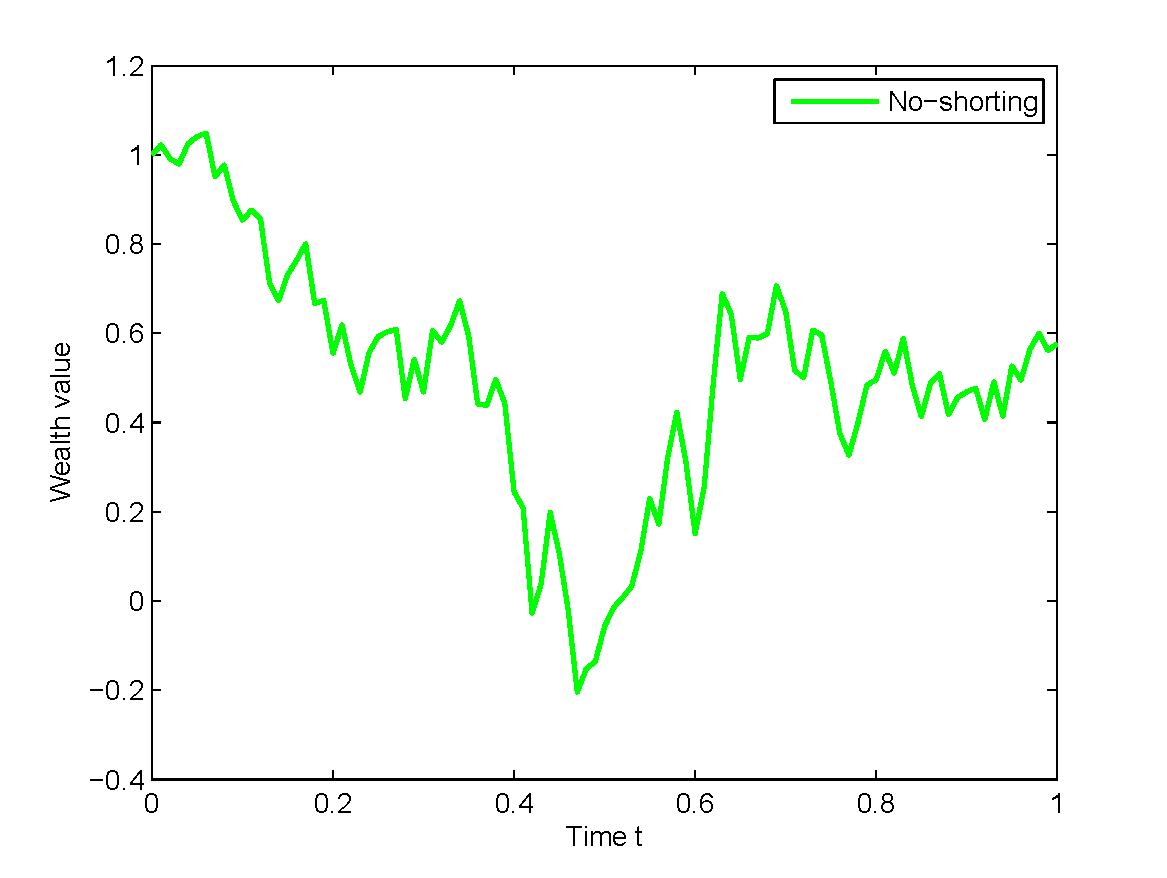} \\
\mbox{Different Wealth Processes}
\end{array}
\end{array}
$$

$$
\begin{array}{ccc}
\begin{array}{c}
\includegraphics[height=2.5in]{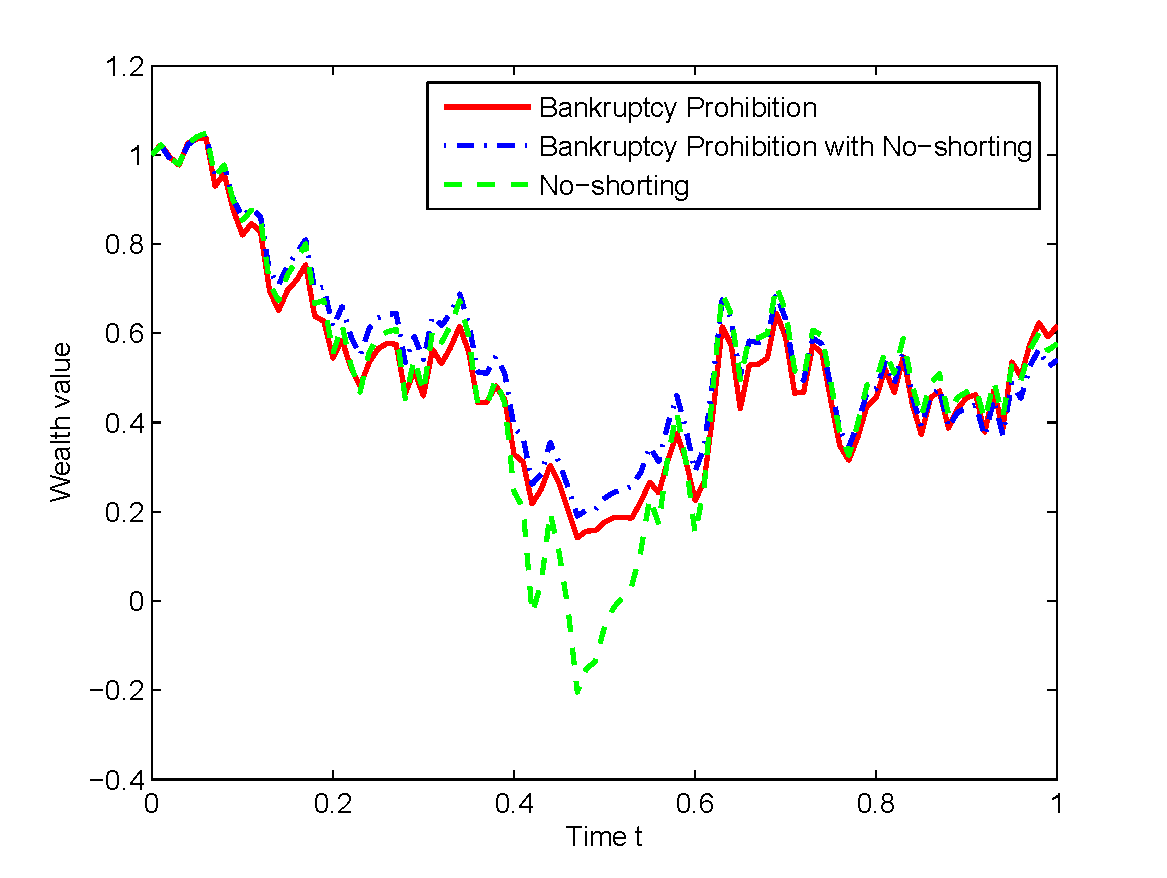} \\
\mbox{Comparision of Wealth Processes}
\end{array}
\end{array}
$$

\section{Conclusion}
\noindent
We have studied the continuous-time mean-variance portfolio selection with mixed restrictions of bankruptcy prohibition (\textit{constrained state})
and convex cone portfolio constraints (\textit{constrained controls}).
The main contribution of the paper is that we developed semi-analytical expressions for the pre-committed efficient mean-variance policy without the viscosity solution technique.
A natural extension of our result to continuous-time linear-quadratic cone constrained controls with constrained states is straightforward, at
least conceptually. On the other hand, if the rates of all market coefficients are random, the problem becomes more complicated.
An important challenge appears when one considers general markets with convex portfolio constraints, so the constraints may not be cone.


\end{document}